\documentclass[journal]{IEEEtran}
\usepackage{cite}
\usepackage{amsmath,amssymb,amsfonts}
\usepackage{algorithmic}
\usepackage{graphicx}
\usepackage{textcomp}
\usepackage{mathtools}				
\usepackage{bm}  					
\usepackage[ruled,linesnumbered]{algorithm2e}
\usepackage{subfigure}
\usepackage{xcolor}
\usepackage{color}
\usepackage{stfloats}
\usepackage{endnotes}
\usepackage{url}

\usepackage{amsthm}
\newtheorem{theorem}{Theorem}
\newtheorem{lemma}{Lemma}

\newtheorem{assumption}{Assumption}

\newcommand{\T}{\top}

\newcommand{\dif}{\ensuremath{\mathrm{d}}}

\usepackage{multirow}
\usepackage{array}
\usepackage{booktabs}
\graphicspath{{figure/}}   

\usepackage{xcolor}

\begin{document}

\title{Iterated Extended Kalman Smoother-Based Variable Splitting for $L_1$-Regularized State Estimation}

\author{Rui Gao,
        Filip Tronarp,
        and Simo~S\"arkk\"a,~\IEEEmembership{Senior Member,~IEEE}
\thanks{
R.~Gao, F. Tronarp and S.~S\"arkk\"a are with the Department of Electrical Engineering and Automation, Aalto University, Espoo, 02150 Finland. E-mail: \{rui.gao, filip.tronarp, simo.sarkka\}@aalto.fi).
}}

\markboth{Journal of \LaTeX\ Class Files,~Vol.~14, No.~8, August~2015}
{Shell \MakeLowercase{\textit{et al.}}: Bare Demo of IEEEtran.cls for IEEE Journals}

\maketitle

\begin{abstract}
In this paper, we propose a new framework for solving state estimation problems with an additional sparsity-promoting $L_1$-regularizer term. We first formulate such problems as minimization of the sum of linear or nonlinear quadratic error terms and an extra regularizer, and then present novel algorithms which solve the linear and nonlinear cases. The methods are based on a combination of the iterated extended Kalman smoother and variable splitting techniques such as alternating direction method of multipliers (ADMM). We present a general algorithmic framework for variable splitting methods, where the iterative steps involving minimization of the nonlinear quadratic terms can be computed efficiently by iterated smoothing. Due to the use of state estimation algorithms, the proposed framework has a low per-iteration time complexity, which makes it suitable for solving a large-scale or high-dimensional state estimation problem. We also provide convergence results for the proposed algorithms. The experiments show the promising performance and speed-ups provided by the methods. 
\end{abstract}


\IEEEpeerreviewmaketitle

\section{Introduction}
\label{sec:itroduction}
\IEEEPARstart{S}{TATE} estimation problems naturally arise in many signal processing applications including target tracking, smart grids, and robotics \cite{Bar-Shalom+Li+Kirubarajan:2001, simo2013Bayesian, Optimal2017Mallada}. In conventional Bayesian approaches, the estimation task is cast as a statistical inverse problem for restoring the original time series from imperfect measurements, based on a statistical model for the measurements given the signal together with a statistical model for the signal. In linear Gaussian models, this problem admits a closed-form solution, which can be efficiently implemented by the Kalman (or Rauch--Tung--Striebel) smoother (KS)~\cite{RTS1965,simo2013Bayesian}. For nonlinear Gaussian models we can use various linearization and sigma-point-based methods \cite{simo2013Bayesian} for approximate inference. In particular, here we use the so-called iterated extended Kalman smoother (IEKS) \cite{Bell1994smoother}, which is based on analytical linearisation of the nonlinear functions. Although the aforementioned smoothers are often used to estimate dynamic signals, they lack a mechanism to promote sparsity in the signals. 

One approach for promoting sparsity is to add an $L_1$-term to the cost function formulation of the state estimation problem. This approach imposes sparsity on the state estimate, which is either based on a \emph{synthesis sparse} or an \emph{analysis sparse} signal model. A synthesis sparse model assumes that the signal can be represented as a linear combination of basis vectors, where the coefficients are subject to, for example, an $L_1$-penalty, thus promoting sparsity. In the past decade, the use of synthesis sparsity for estimating dynamic signals has drawn a lot of attention \cite{Tropp2010sparse, Carmi2010pseudo, Vaswani2008compressed, Zachariah2012iterative, Farahmand2011sparsity, Aravkin2011Laplace, Aravkin2017Generalized,Simonetto2017Prediction,Charles2011Sparsity,Charles2016Dynamic}. For example,  a pseudo-measurement technique was used in the Kalman update equations for encouraging sparse solutions \cite{Carmi2010pseudo}.  A method based on sparsity was applied compressive sensing to update Kalman innovations or filtering errors \cite{Vaswani2008compressed}.  Based on synthesis sparsity, the estimation problem has been formulated as an $L_1$-regularized least square problem in \cite{Charles2011Sparsity}. Nevertheless, the previously mentioned methods only consider synthesis sparsity of the signal and assume a linear dynamic system. 

On the other hand, analysis sparsity, also called cosparsity, assumes that the signal is not sparse itself, but rather the outcome is sparse or compressible in some transform domain, which leads to the flexibility in the modeling of signals \cite{Elad2007analysis,Gao2017analysis,Yaghoobi2013cosparse, Rubinstein2013cosparse,Turek2014analysis}.  Analysis sparse models involving an analysis operator -- a popular choice being total variation (TV) -- have been very successful in image processing. 
For example,  several algorithms \cite{Yaghoobi2013cosparse,Rubinstein2013cosparse} have been developed to train an analysis operator and the trained operators have been used for image denoising. In \cite{Hu2012TV} the authors proposed to use the TV regularizer to improve the quality of image reconstruction.  However, these approaches are not ideally suited for reconstructing dynamic signals. In state estimation problems, the available methods for analysis sparse priors are still limited. The main goal of this paper is to introduce these kinds of methods for dynamic state estimation.

Formulating a state estimation problem using synthesis and analysis sparsity leads to a general class of optimization problems, which require minimization of composite functions such as an analysis-$L_1$-regularized least-square problems. The difficulties arise from the appearance of the nonsmooth regularizer. There are various batch optimization methods such as proximal gradient method \cite{Chalasani2014proximal}, Douglas-Rachford splitting (DRS) \cite{Douglas1992,Ryu2016}, Peaceman-Rachford splitting (PRS) \cite{Peaceman1955,He2014peaceman}, the split Bregman method (SBM)~\cite{Goldstein2009Split}, the alternating method of multipliers (ADMM)~\cite{Boyd2011admm,Glowinski2014ADMM}, and the first-order primal-dual (FOPD) method~\cite{Chambolle2011fopd} for addressing this problem. 
However, these general methods do not take the inherent temporal nature of the optimization problem into account, which leads to bad computational and memory scaling in large-scale or high-dimensional data. This often renders the existing methods intractable due to their extensive memory and computational requirements. 

As a consequence, we propose to combine a Kalman smoother with variable splitting optimization methods, which allows us to account for the temporal nature of the data in order to speed up the computations. In this paper, we derive novel methods for efficiently estimating dynamic signals with an extra (analysis) $L_1$-regularized term. The developed algorithms are based on using computationally efficient KS and IEKS for solving the subproblems arising within the steps of the optimization methods. Our experiments demonstrate promising performance of the methods in practical applications.  
The main contributions are as follows:

\begin{enumerate}
\item We formulate the state estimation problem as an optimization problem that is based upon a general sparse model containing analysis or synthesis prior. The formulation accommodates a large class of popular sparsifying regularizers (e.g., synthesis $L_1$-norm, analysis $L_1$-norm, TV norm) for state estimation. 

\item We present novel practical optimization methods, KS-ADMM and IEKS-ADMM, which are based on combining ADMM with KS and IEKS, respectively. 

\item We also prove the convergence of the KS-ADMM method as well as the local convergence of the IEKS-ADMM method. 

\item We generalize our smoother-based approaches to a general class of variable splitting techniques. 
\end{enumerate}

The advantage of the proposed approach is that the computational cost per iteration is much less than in the conventional batch solutions. Our approach is computationally superior to the state-of-the-art in a large-scale or high-dimensional state estimation applications. 

The rest of the paper is organized as follows. We conclude this section by reviewing variable splitting methods and IEKS. Section \ref{sec:KS-ADMM} first develops the batch optimization by a classical ADMM method, and then presents a new KS-ADMM method for solving a linear dynamic estimation problem. Furthermore, for the nonlinear case, we present an IEKS-ADMM method in Section \ref{sec:IEKS-ADMM} and establish its local convergence properties. Section \ref{sec:framework} introduces a more general smoother-based variable splitting algorithmic framework. In particular, a general IEKS-based optimization method is formulated. Various experimental results in Section \ref{sec:Experiments} demonstrate the effectiveness and accuracy in simulated linear and nonlinear state estimation problem. The performance of the algorithm is also illustrated in real-world tomographic reconstruction.

The notation of the paper is as follows. Vectors $\mathbf{x}$ and matrices $\mathbf{X}$ are indicated in boldface. $(\cdot)^{\T}$ stands for transposition, and $(\cdot)^{-1}$ is the matrix inversion. $\mathbf{x}_{1:T}$ stands for the time series from $\mathbf{x}_1$ to $\mathbf{x}_T$, and $\mathbf{x}^{(k)}$ denotes the value of $\mathbf{x}$ at $k$:th iteration. $\langle \mathbf{x}, \mathbf{y} \rangle$ represents the vector inner product $\mathbf{x}^\T\mathbf{y}$. We denote by $\mathbb{R}^{n}$ the usual $n$ dimensional Euclidean space. The vector norm $\left \|\cdot\right \|_p$ for $p \geq 1$ is the standard $\ell_p$-norm. The $\mathbf{R}$-weighted Euclidean norm of a vector $\mathbf{x}$ is denoted by $\|\mathbf{x}\|_{\mathbf{R}}= \sqrt{\mathbf{x}^\T \mathbf{R} \mathbf{x}}$. $\theta^*$ is the conjugate of a convex function $\theta$, defined as $\theta^*(\mathbf{p}) = \mathop{\sup}_{\mathbf{x}} \langle \mathbf{x}, \mathbf{p} \rangle - \theta (\mathbf{x})$. $\operatorname{sgn}$ represents the signum function. $ \operatorname{vec}(\cdot)$ is the vectorization operator, $\operatorname{blkdiag}(\cdot)$ is a block diagonal matrix operator with the elements in its argument on the diagonal, $\partial \phi(\mathbf{x})$ denotes a subgradient of $\phi$ at $\mathbf{x}$, $\mathbf{J}_\phi$ is the Jacobian of $\phi(\mathbf{x})$ and $\nabla \phi(\mathbf{x})$ and $\nabla^2 \phi(\mathbf{x})$ are the gradient and Hessian of the function $\phi(\mathbf{x})$.

\subsection{Problem Formulation}
\label{sec:problem_formulation}
Consider the dynamic state-space model \cite{Bar-Shalom+Li+Kirubarajan:2001, simo2013Bayesian}
\begin{equation}\label{eq:model}
\begin{split}
\begin{aligned}
\mathbf{x}_{t} &=\mathbf{a}_t(\mathbf{x}_{t-1})+ \mathbf{q}_t, \\
\mathbf{y}_{t} &=\mathbf{h}_t(\mathbf{x}_t)+ \mathbf{r}_t, 
\end{aligned}
\end{split}
\end{equation}

where $t=1,\ldots,T$, $\mathbf{x}_{t} = \begin{bmatrix} x_{1,t}& x_{2,t}& \ldots& x_{n_x,t} \end{bmatrix}^\T \in \mathbb{R}^{n_x}$ denotes an $n_x$-dimensional state of the system at the time step $t$, and $\mathbf{y}_{t} = \begin{bmatrix} y_{1,t}& y_{2,t}& \ldots& y_{n_y,t}\end{bmatrix}^\T \in \mathbb{R}^{n_y}$ is an $n_y$-dimensional noisy measurement signal, $\mathbf{h}_t: \mathbb{R}^{n_x} \to \mathbb{R}^{n_y}$ is a measurement function (typically with $n_y \leq n_x$), and $\mathbf{a}_t: \mathbb{R}^{n_x} \to \mathbb{R}^{n_x}$ is a state transition function at time step $t$. The initial state $\mathbf{x}_1$ is assumed to have mean $\mathbf{m}_1$ and covariance $\mathbf{P}_1$. The errors $\mathbf{q}_t$ and $\mathbf{r}_t$ are assumed to be mutually independent random variables with known positive definite covariance matrices $\mathbf{Q}_t$ and $\mathbf{R}_t$, respectively. The goal is to estimate the state sequence $\mathbf{x}_{1:{T}} = \left\{\mathbf{x}_1,\ldots,\mathbf{x}_{T}\right\}$ from the noisy measurement sequence $\mathbf{y}_{1:{T}} = \{\mathbf{y}_1, \ldots,\mathbf{y}_{T}\}$. In this paper, we focus on estimating $\mathbf{x}_{1:{T}} $ by minimizing the sum of quadratic error terms and an $L_1$ sparsity-promoting penalty. 

For the sparsity assumption, we add an extra $L_1$-penalty for the state $\mathbf{x}_t$, and then formulate the optimization problem as
\begin{equation}\label{eq:optimization}
\begin{split}
\begin{aligned}
&\mathbf{x}^\star_{1:{T}}  = \mathop{\arg\min}_{\mathbf{x}_{1:{T}}}
 \frac{1}{2} \sum_{t=1}^{T}  \| \mathbf{y}_t - \mathbf{h}_t(\mathbf{x}_t)  \|_{\mathbf{R}_t^{-1}}^2 + 
 \lambda \sum_{t=1}^{T} \left \| \mathbf{\Omega}_t\mathbf{x}_t \right \|_1 \\
 &\quad + \frac{1}{2}\sum_{t=2}^{{T}} \|\mathbf{x}_{t}-\mathbf{a}_t(\mathbf{x}_{t-1})\|_{ \mathbf{Q}_t^{-1} }^2  
+\frac{1}{2} \| \mathbf{x}_1  -   \mathbf{m}_1 \|_{ \mathbf{P}_1^{-1} }^2,
\end{aligned}
\end{split}
\end{equation}
where $\mathbf{x}^\star_{1:{T}} $ is the optimal state sequence, $\mathbf{\Omega}_t$ is a linear operator, and $\lambda$ is a penalty parameter, which describes a trade-off between the data fidelity term and the regularizing penalty term. The formulation \eqref{eq:optimization} encompasses two particular cases: by setting 
$\mathbf{\Omega}_t$ to a diagonal matrix (e.g., identity matrix $\mathbf{\Omega}_t = \mathbf{I}$), a synthesis sparse model is obtained, which assumes that $\mathbf{x}_{1:T}$ are sparse. Such a case arises frequently in state estimation applications \cite{Charles2016Dynamic, Farahmand2011sparsity, Aravkin2011Laplace, Ziniel2013Dynamic, Aravkin2017Generalized}. Correspondingly, an analysis sparse model is obtained when a more general $\mathbf{\Omega}_t$ is used. For example, the TV regularization, which is common in tomographic reconstruction, can be obtained by using a finite-difference matrix as $\mathbf{\Omega}_t$. 

More generally, $\mathbf{\Omega}_t$ can be a fixed matrix \cite{Shen2006tight,Turek2014analysis,Elad2007analysis} or a learned matrix \cite{Gao2017analysis,Yaghoobi2013cosparse, Rubinstein2013cosparse}. It should be noted that, if the $L_1$ term is not used (i.e., when $\lambda = 0$) in \eqref{eq:optimization}, the objective can be solved by using a linear or non-linear KS \cite{RTS1965,Bell1994smoother,simo2013Bayesian}. However, when $\lambda > 0$, the smoothing is no longer applicable, and the cost function is non-differentiable.

Since $\| \mathbf{\Omega}_t\mathbf{x}_t\|_1$ does not have a closed-form proximal operator in general, we employ variable splitting technique for solving the resulting optimization problem. As mentioned above, many variable splitting methods can be used to solve \eqref{eq:optimization}, such as PRS \cite{He2014peaceman}, SBM~\cite{Goldstein2009Split}, ADMM~\cite{Boyd2011admm},  DRS \cite{Douglas1992}, and FOPD~\cite{Chambolle2011fopd}. Especially, ADMM is a popular member of this class. Therefore, we start by presenting algorithms based on ADMM and then extend them to more general variable splitting methods. In the following, we review variable splitting and IEKS methods, before presenting our approach in detail.

\subsection{Variable Splitting}
\label{sec:various_methods}
The methods we develop in this paper are based on variable splitting \cite{Courant1943Variational,Wang2008vs}. Consider an unconstrained optimization problem in which the objective function is the sum of two functions 
\begin{equation}\label{eq:general_function}
\begin{split}
\begin{aligned}
\min_{\mathbf{x}} \theta_1(\mathbf{x}) + \theta_2(\mathbf{\Omega} \mathbf{x}), 
\end{aligned}
\end{split}
\end{equation}
where $\theta_2(\cdot) = \| \cdot \|_1$, and $\mathbf{\Omega}$ is a matrix. 
Variable splitting refers to the process of introducing an auxiliary constrained variable $\mathbf{w}$ to separate the components in the cost function. More specifically, 
we impose the constraint $\mathbf{w} = \mathbf{\Omega}\mathbf{x}$, which transforms the original minimization problem \eqref{eq:general_function} into an equivalent
constrained minimization problem, given by
\begin{equation}\label{eq:general_constrained_function}
\begin{split}
\begin{aligned}
\min_{\mathbf{x},\mathbf{w}} \theta_1(\mathbf{x}) + \theta_2(\mathbf{w}),  \quad \mathrm{s.t.}\quad \mathbf{w} = \mathbf{\Omega}\mathbf{x}.
\end{aligned}
\end{split}
\end{equation}
The minimization problem \eqref{eq:general_constrained_function}  can be solved efficiently by classical constrained optimization methods \cite{Wright2006Numerical}. The rationale of variable splitting is that it may be easier to solve the constrained problem \eqref{eq:general_constrained_function} than the unconstrained one \eqref{eq:general_function}. PRS, SBM, FOPD, ADMM, and their variants \cite{Ouyang2013} are a few well-known variable splitting methods -- see also \cite{Ryu2016monotone,Esser2009Applications} for a recent historical overview. 

ADMM \cite{Boyd2011admm} is one of the most popular algorithms for solving \eqref{eq:general_constrained_function}. ADMM defines an augmented Lagrangian function, and then alternates between the updates of the split variables. Given $\mathbf{x}^{(0)}$, $\mathbf{w}^{(0)}$, and $ \bm{\eta}^{(0)}$, its iterative steps are:
\begin{subequations}
\label{eq:general-admm}
\begin{align}  
\label{eq:x-primal-admm}
\mathbf{x}^{(k+1)} &= \mathop{\arg\min}_{\mathbf{x}}  \theta_1(\mathbf{x}) + (\bm{\eta}^{(k)})^\T ( \mathbf{w}^{(k)} - \mathbf{\Omega}\mathbf{x})  \notag \\  
& \quad \quad \quad + \frac{\rho}{2} \|\mathbf{w}^{(k)} - \mathbf{\Omega}\mathbf{x} \|^2, \\
\mathbf{w}^{(k+1)} &= \mathop{\arg\min}_{\mathbf{w}}  \theta_2(\mathbf{w}) + (\bm{\eta}^{(k)})^\T (\mathbf{w} - \mathbf{\Omega}\mathbf{x}^{(k+1)} ) \notag \\
& \quad \quad \quad + \frac{\rho}{2} \|  \mathbf{w} - \mathbf{\Omega}\mathbf{x}^{(k+1)} \|^2, \\
\bm{\eta}^{(k+1)} &= \bm{\eta}^{(k)} + \rho ( \mathbf{w}^{(k+1)} - \mathbf{\Omega}\mathbf{x}^{(k+1)}),
\end{align}
\end{subequations}
where $\bm{\eta}$ is a Lagrange multiplier and $\rho$ is a parameter. 

The PRS method \cite{Peaceman1955,He2014peaceman} is similar to ADMM except that it updates the Lagrange multiplier twice. The typical iterative steps for~\eqref{eq:general_function} are
\begin{subequations}
\begin{align}  
\label{eq:x-primal-PRS}
\mathbf{x}^{(k+1)} &= \mathop{\arg\min}_{\mathbf{x}}  \theta_1(\mathbf{x}) + (\bm{\eta}^{(k)})^\T ( \mathbf{w}^{(k)} - \mathbf{\Omega}\mathbf{x} )  \notag \\  
& \quad \quad \quad + \frac{\rho}{2} \| \mathbf{w}^{(k) }- \mathbf{\Omega}\mathbf{x} \|^2, \\
\bm{\eta}^{(k+\frac{1}{2})} &= \bm{\eta}^{(k)} + \alpha \rho (\mathbf{w}^{(k)} - \mathbf{\Omega}\mathbf{x}^{(k+1)}), \\
\mathbf{w}^{(k+1)} &=\mathop{\arg\min}_{\mathbf{w}}  \theta_2(\mathbf{w}) + (\bm{\eta}^{(k+\frac{1}{2})})^\T (  \mathbf{w} - \mathbf{\Omega}\mathbf{x}^{(k+1)}) \notag \\
& \quad \quad \quad + \frac{\rho}{2} \|  \mathbf{w} - \mathbf{\Omega}\mathbf{x}^{(k+1)} \|^2, \\
\bm{\eta}^{(k+1)} &= \bm{\eta}^{(k+\frac{1}{2})} + \alpha \rho (\mathbf{w}^{(k+1)} - \mathbf{\Omega}\mathbf{x}^{(k+1)}),
\end{align}
\end{subequations}
where $\alpha \in(0,1)$. 

In SBM~\cite{Goldstein2009Split}, we iterate the steps
\begin{subequations}
\begin{align}  
\label{eq:x-primal-SBM}
&\mathbf{x}^{(k+1)} = \mathop{\arg\min}_{\mathbf{x}}  \theta_1 (\mathbf{x})
+ \frac{\rho}{2} \|\mathbf{w}^{(k)} - \mathbf{\Omega} \mathbf{x}  +  \bm{\eta}^{(k)} \|_2^2, \\
& \mathbf{w}^{(k+1)}  = \mathop{\arg\min}_{\mathbf{w}}  \theta_2(\mathbf{w})  +  \frac{\rho}{2} \|\mathbf{w} - \mathbf{\Omega} \mathbf{x}^{(k+1)} +  \bm{\eta}^{(k)} \|_2^2,
\end{align}
\end{subequations}
$M$ times, and update the extra variable by
\begin{equation}\label{eq:sbm_eta}
\begin{split}
\begin{aligned}
\bm{\eta}^{(k+1)} = \bm{\eta}^{(k)}+ ( \mathbf{w}^{(k+1)} - \mathbf{\Omega}\mathbf{x}^{(k+1)}).
\end{aligned}
\end{split}
\end{equation}
When $M=1$, this is equivalent to ADMM.

There are also other variable splitting methods which alternate proximal steps for the primal and dual variables. One example is FOPD~\cite{Chambolle2011fopd}, where the $(k+1)$:th iteration consists of the following 
{\small
\begin{subequations}
\begin{align}  
& \mathbf{w}^{(k+1)} = \mathop{\arg\min}_{\mathbf{w}}  \theta_2^*(\mathbf{w})  +  \frac{1}{2 \gamma} \|\mathbf{w} - ( \mathbf{w}^{(k)} + \gamma \mathbf{\Omega} \hat{\mathbf{x}}^{(k)} ) \|^2, \\
\label{eq:x-primal-FOPD}
& \mathbf{x}^{(k+1)} =  \mathop{\arg\min}_{\mathbf{x}}  \theta_1(\mathbf{x}) +   \frac{1}{2\rho} \|\mathbf{x}   -  ( \mathbf{x}^{(k)} \!  -  \!   \rho \mathbf{\Omega}^\T  \mathbf{w}^{(k+1)} )\|^2, \\
&\hat{\mathbf{x}}^{(k+1)}  = \mathbf{x}^{(k+1)}  + \tau (\mathbf{x}^{(k+1)}  - {\mathbf{x}}^{(k)} ),
\end{align}
\end{subequations}
}
\noindent where $\tau$ and $\gamma$ are parameters. 

All these variable splitting algorithms provide simple ways to construct efficient iterative algorithms that offer simpler inner subproblems. However, the subproblems such as \eqref{eq:x-primal-admm}, \eqref{eq:x-primal-PRS}, \eqref{eq:x-primal-SBM} and \eqref{eq:x-primal-FOPD} remain computationally expensive, as they involve large matrix-vector products when the dimensionality of $\mathbf{x}$ is large. We circumvent this problem by combining variable splitting with KS and IEKS.

\subsection{The Iterated Extended Kalman Smoother}
\label{sec:ieks}

IEKS \cite{Bell1994smoother} is an approximative algorithm for solving non-linear optimal smoothing problems. However, it can also be seen as an efficient implementation of the Gauss--Newton algorithm for solving the problem
\begin{equation}\label{eq:ieks_optimization}
\begin{split}
\begin{aligned}
&\mathbf{x}^\star_{1:{T}}  = \mathop{\arg\min}_{\mathbf{x}_{1:{T}}}
 \frac{1}{2} \sum_{t=1}^{T}  \| \mathbf{y}_t - \mathbf{h}_t(\mathbf{x}_t)  \|_{\mathbf{R}_t^{-1}}^2  \\
&+ \frac{1}{2}\sum_{t=2}^{{T}} \|\mathbf{x}_{t}-\mathbf{a}_t(\mathbf{x}_{t-1})\|_{ \mathbf{Q}_t^{-1} }^2
+ \frac{1}{2} \| \mathbf{x}_1  -   \mathbf{m}_1 \|_{ \mathbf{P}_1^{-1} }^2.
\end{aligned}
\end{split}
\end{equation}
That is, it produces the maximum a posteriori (MAP) estimate of the trajectory. The IEKS method works by alternating between linearisation of $\mathbf{a}_t$ and $\mathbf{h}_t$ around a previous estimate $\mathbf{x}^{(i)}_{1:{T}}$, as follows:
\begin{subequations}
\label{eq:nonliear_appro}
\begin{align}
\mathbf{a}_t(\mathbf{x}_{t-1}) &\approx \mathbf{a}_t(\mathbf{x}_{t-1}^{(i)}) + \mathbf{J}_{a_t}(\mathbf{x}_{t-1}^{(i)}) (\mathbf{x}_{t-1} - \mathbf{x}_{t-1}^{(i)}), \\
\mathbf{h}_t(\mathbf{x}_t) &\approx \mathbf{h}_t(\mathbf{x}_t^{(i)}) + \mathbf{J}_{h_t}(\mathbf{x}_t^{(i)}) (\mathbf{x}_t - \mathbf{x}_t^{(i)}),
\end{align}
\end{subequations}
and solving the linearized problem

\begin{equation}
\label{eq:ieks_linerisation}
\begin{split}
\begin{aligned}
&\mathbf{x}^{(i+1)}_{1:{T}}  = \\
&\mathop{\arg\min}_{\mathbf{x}_{1:{T}}}
 \frac{1}{2} \sum_{t=1}^{T}  \| \mathbf{y}_t - \mathbf{h}_t(\mathbf{x}_t^{(i)}) - \mathbf{J}_{h_t}(\mathbf{x}_t^{(i)}) (\mathbf{x}_t - \mathbf{x}_t^{(i)})  \|_{\mathbf{R}_t^{-1}}^2  \\
&  + \frac{1}{2}\sum_{t=2}^{{T}} \|\mathbf{x}_{t}-\mathbf{a}_t(\mathbf{x}_{t-1}^{(i)}) - \mathbf{J}_{a_t}(\mathbf{x}_{t-1}^{(i)})  (\mathbf{x}_t - \mathbf{x}_t^{(i)}) \|_{ \mathbf{Q}_t^{-1} }^2  \\
&  +  \frac{1}{2} \| \mathbf{x}_1  -   \mathbf{m}_1 \|_{ \mathbf{P}_1^{-1} }^2.
\end{aligned}
\end{split}
\end{equation}
The solution of \eqref{eq:ieks_linerisation} can in turn be efficiently obtained by the Rauch--Tung--Striebel (RTS) smoother \cite{RTS1965}, which first computes the filtering mean and covariances $\mathbf{m}_{1:T}$ and $\mathbf{P}_{1:T}$, by alternating between prediction
\begin{subequations}
\label{eq:nonliear_pre}
\begin{align}
\mathbf{m}_t^- &= \mathbf{a}_t(\mathbf{x}_{t-1}^{(i)}) + \mathbf{J}_{a_t}(\mathbf{x}_{t-1}^{(i)}) (\mathbf{m}_{t-1} - \mathbf{x}_{t-1}^{(i)}), \\
\mathbf{P}_t^- &= \mathbf{J}_{a_t}(\mathbf{x}_{t-1}^{(i)}) \mathbf{P}_{t-1} [\mathbf{J}_{a_t}(\mathbf{x}_{t-1}^{(i)})]^\T + \mathbf{Q}_t,
\end{align}
\end{subequations} 
and update
\begin{subequations}
\label{eq:nonliear_pre_next}
\begin{align}
&\mathbf{S}_t =   \mathbf{J}_{h_t}(\mathbf{x}_t^{(i)}) \mathbf{P}_t^- [\mathbf{J}_{h_t}(\mathbf{x}_t^{(i)})]^\T  + \mathbf{R}_t, \\
&\mathbf{K}_t = \mathbf{P}_t^- [\mathbf{J}_{h_t}(\mathbf{x}_t^{(i)})]^\T [\mathbf{S}_t]^{-1},\\
&\mathbf{m}_t = \mathbf{m}_t^-  + \mathbf{K}_t \Big(\mathbf{y}_t - \mathbf{h}_t(\mathbf{x}_t^{(i)}) - \mathbf{J}_{h_t}(\mathbf{x}_t^{(i)}) (\mathbf{m}_t^- - \mathbf{x}_t^{(i)}) \Big), \\
&\mathbf{P}_t = \mathbf{P}_t^- - \mathbf{K}_t \mathbf{S}_t  [\mathbf{K}_t ]^\T,
\end{align}
\end{subequations} 
where $\mathbf{S}_t$ and $\mathbf{K}_t$ are the innovation covariance matrix and the Kalman gain at the time step $t$, respectively. The filtering means $\mathbf{m}_t$ and covariances $\mathbf{P}_t$ are then corrected in a backwards (smoothing) pass
\begin{subequations}
\label{eq:nonliear_smoother}
\begin{align}
\mathbf{G}_t &= \mathbf{P}_t [\mathbf{J}_{a_t}(\mathbf{x}_{t-1}^{(i)})]^\T [\mathbf{P}_{t+1}^-]^{-1}, \\
\mathbf{m}_t &= \mathbf{m}_t  + \mathbf{G}_t \Big(\mathbf{m}_{t+1}^s - \mathbf{m}_{t+1}^- \Big), \\
\mathbf{P}_t^s &= \mathbf{P}_t + \mathbf{G}_t \Big( \mathbf{P}_{t+1}^s - \mathbf{P}_{t+1}^- \Big)[\mathbf{G}_t ]^\T. 
\end{align}
\end{subequations}
Now setting $\mathbf{x}^{(i+1)}_t = \mathbf{m}_t^s$ gives the solution to \eqref{eq:ieks_linerisation}. When the functions $\mathbf{a}_t$ and $\mathbf{h}_t$ are linear, the above iteration converges in a single step. This algorithm is the classical RTS smoother or more briefly KS \cite{RTS1965}.

In this paper, we use the KS and IEKS algorithms as efficient methods for solving generalized versions of the optimization problems given in \eqref{eq:ieks_optimization}, which arise within the steps of variable splitting.

\vspace{-5pt}
\section{Linear State Estimation by KS-ADMM}
\label{sec:KS-ADMM}
In this section, we present the KS-ADMM algorithm which is a novel algorithm for solving $L_1$-regularized linear Gaussian state estimation problems. In particular, Section~\ref{sec:Optimization_linear} describes the batch solution by ADMM. Then, by defining an artificial measurement noise and a pseudo-measurement, we formulate the KS-ADMM algorithm to solve the primal variable update in Section \ref{sec:x_linear}. 

\subsection{Batch Optimization}
\label{sec:Optimization_linear}
Let us assume that the state transition function $\mathbf{a}_t$ and the measurement function $\mathbf{h}_t$ are linear, denoted by
\begin{equation}\label{eq:x}
\begin{split}
\begin{aligned} 
\mathbf{a}_t(\mathbf{x}_{t-1}) &= \mathbf{A}_t\mathbf{x}_{t-1}, \\
\mathbf{h}_t(\mathbf{x}_{t}) &= \mathbf{H}_t\mathbf{x}_{t}, 
	\end{aligned}
	\end{split}
\end{equation}
where $\mathbf{A}_t$ and $\mathbf{H}_t$ are the transition matrix and the measurement matrix, respectively. 
In order to reduce this problem to \eqref{eq:general_function}, we stack the entire state sequence into a vector, which transforms the objective into a batch optimization problem. Thus, we define the following variables
\begin{subequations}
	\label{eq:linear_vector_sets}
	\begin{align}
\mathbf{x} &= \operatorname{vec}(\mathbf{x}_1, \mathbf{x}_2,\ldots,\mathbf{x}_{T}),\\
\mathbf{y} &= \operatorname{vec}(\mathbf{y}_1, \mathbf{y}_2, \ldots,\mathbf{y}_{T}),\\
\mathbf{m} &=  \operatorname{vec}(\mathbf{m}_1, \mathbf{0},\ldots,\mathbf{0}),\\  
\mathbf{H}  &= \operatorname{blkdiag}(\mathbf{H}_1,\mathbf{H}_2,\ldots,\mathbf{H}_{T}), \\
\mathbf{Q}  &= \operatorname{blkdiag}(\mathbf{P}_1,\mathbf{Q}_2,\ldots,\mathbf{Q}_{T}), \\
\mathbf{R}  &= \operatorname{blkdiag}(\mathbf{R}_1,\mathbf{R}_2,\ldots,\mathbf{R}_{T}),\\
\mathbf{\Omega} &= \operatorname{blkdiag}(\mathbf{\Omega}_1,\mathbf{\Omega}_2,\ldots,\mathbf{\Omega}_{T}),  \\
  \mathbf{\Psi} &= \begin{pmatrix}
      \mathbf{I}   & \mathbf{0} & & \\
    -\mathbf{A}_2  &       \mathbf{I}  &  \ddots & \\
   & \ddots & \ddots & \mathbf{0} \\
   & &  -\mathbf{A}_T &  \mathbf{I} \\
  \end{pmatrix}.
\end{align}
\end{subequations}

The optimization problem introduced in Section~\ref{sec:problem_formulation} can now be reformulated as the following batch optimization problem
\begin{equation}\label{eq:linear_whole_optimization} 
\begin{split}
\begin{aligned}
\mathbf{x}^\star  &= \mathop{\arg\min}_{\mathbf{x}} 
\frac{1}{2} \left \| \mathbf{y} - \mathbf{H}\mathbf{x}\right \|_{\mathbf{R}^{-1}}^2 
+ \frac{1}{2} \left \| \mathbf{\Psi}\mathbf{x}  - \mathbf{m} \right \|_{\mathbf{Q}^{-1}}^2  \\
&\quad + \lambda \left \|  \mathbf{\Omega} \mathbf{x} \right \|_1,
\end{aligned}
\end{split}
\end{equation}
which in turn can be seen to be a special case of \eqref{eq:general_function}. Here, our algorithm for solving \eqref{eq:linear_whole_optimization} builds upon the batch ADMM~\cite{Boyd2011admm}. 

To derive an ADMM algorithm for \eqref{eq:linear_whole_optimization}, we introduce an auxiliary variable $\mathbf{w} = \operatorname{vec}(\mathbf{w}_1, \ldots,\mathbf{w}_{T})$ and a linear equality constraint $\mathbf{w} = \mathbf{\Omega}\mathbf{x}$. The resulting equality-constrained problem is formulated mathematically as 
\begin{equation}\label{eq:admm_linear} 
\begin{split}
\begin{aligned}
\min_{\mathbf{x}}\ &\frac{1}{2}\| \mathbf{y} - \mathbf{H}\mathbf{x} \|_{\mathbf{R}^{-1}}^2
 + \frac{1}{2}  \|\mathbf{\Psi}\mathbf{x}  - \mathbf{m} \|_{ \mathbf{Q}^{-1}}^2  
 +\lambda  \| \mathbf{w}  \|_1 \\
{\mathrm{s.t.}}\ &\mathbf{w} = \mathbf{\Omega}\mathbf{x}.
\end{aligned}
\end{split}
\end{equation}
The main objective here is to find a stationary point $(\mathbf{x}^\star, \mathbf{w}^\star, \bm{\eta}^\star)$ of the augmented Lagrangian function associated with \eqref{eq:admm_linear} as the function
\begin{equation}\label{eq:Lagrangian_func}
\begin{split}
\begin{aligned}
&\mathcal{L}(\mathbf{x},\mathbf{w};\bm{\eta}) \triangleq  
\frac{1}{2}  \| \mathbf{y} - \mathbf{H}\mathbf{x} \|_{\mathbf{R}^{-1}}^2 
+ \lambda \| \mathbf{w}  \|_1 \\ 
& + \frac{1}{2} \|\mathbf{\Psi}\mathbf{x}  - \mathbf{m}\|_{ \mathbf{Q}^{-1} }^2  
 + \bm{\eta}^\T(\mathbf{w}- \mathbf{\Omega}\mathbf{x})
+ \frac{\rho}{2} \| \mathbf{w}- \mathbf{\Omega}\mathbf{x}\|^2,
\end{aligned}
\end{split}
\end{equation}
where $\bm{\eta} \in \mathbb{R}^{T P}$ is the dual variable and $\rho$ is a penalty parameter. As described in Section \ref{sec:various_methods}, at each iteration of ADMM we perform the updates 
 \begin{subequations}
\label{eq:admm_step}
\begin{align} \label{eq:admm_x_linear}
       \mathbf{x}^{(k+1)} &= \mathop{\arg\min}_\mathbf{x}\mathcal{L}(\mathbf{x},\mathbf{w}^{(k)};\bm{\eta}^{(k)}), \\ 
       \label{eq:admm_iteration_w}
    \mathbf{w}^{(k+1)} &= \mathop{\arg\min}_\mathbf{w}\mathcal{L}(\mathbf{x}^{(k+1)},\mathbf{w};\bm{\eta}^{(k)}), \\ 
    \label{eq:admm_iteration_eta}
    \bm{\eta}^{(k+1)}&= \bm{\eta}^{(k)}+ \rho( \mathbf{w}^{(k+1)} - \mathbf{\Omega}\mathbf{x}^{(k+1)}).
\end{align}
\end{subequations}
The update for the primal sequence $\mathbf{x}$ is equivalent to the quadratic optimization problem given by
\begin{equation}\label{eq:x_linear}
\begin{split}
\mathbf{x}^{(k+1)}  &= \mathop{\arg\min}_{\mathbf{x}}  \frac{1}{2}   \| \mathbf{y} -  \mathbf{H}\mathbf{x} \|_{\mathbf{R}^{-1}}^2 
+ \frac{1}{2} \|\mathbf{\Psi}\mathbf{x}  - \mathbf{m}\|_{ \mathbf{Q}^{-1} }^2  \\ 
& \qquad + \frac{\rho}{2} \| \mathbf{w}- \mathbf{\Omega}\mathbf{x} + \bm{\eta}/ \rho\|^2,
\end{split}
\end{equation}
which has the closed-form solution
\begin{equation}
\label{eq:x_linear_solution}
\begin{split}
&  \mathbf{x}^{(k+1)}  = \left[ \mathbf{H}^\T \mathbf{R}^{-1} \mathbf{H} 
   +  \mathbf{\Psi}^\T \mathbf{Q}^{-1} \mathbf{\Psi}
   + \rho  \mathbf{\Omega}^\T\mathbf{\Omega} \right]^{-1} \\
   & \times \left[ \mathbf{H}^\T \mathbf{R}^{-1} \mathbf{y} + 
   \mathbf{\Psi}^\T \mathbf{Q}^{-1} \mathbf{m} + \rho \mathbf{\Omega}^\T ( \mathbf{w}^{(k)} + \bm{\eta}^{(k)}/\rho )  \right].
\end{split}
\end{equation}

For the dual sequence $\mathbf{w}$, the iteration in \eqref{eq:admm_iteration_w} can be solved by \cite{Parikh2013Proxi}
\begin{equation}
\begin{split}
\begin{aligned}
\label{eq:w_compute_2}
&\mathbf{w}^{(k+1)} =
  \operatorname{max} (| \mathbf{e}^{(k)}|- \lambda / \rho,0)  \operatorname{sgn}(\mathbf{e}^{(k)} ) ,
\end{aligned}
\end{split}
\end{equation}
where $\mathbf{e}^{(k)}=\mathbf{\Omega}\mathbf{x}^{(k+1)} + \bm{\eta}^{(k)}/\rho$.

While the optimization problem \eqref{eq:x_linear} can be solved in closed-form, direct solution is computationally demanding, especially when the number of time points or the dimensionality of the state is large. However, the problem can be recognized to be a special case of optimization problems where the iterations can be solved by KS (see Section~\ref{sec:ieks}) provided that we add pseudo-measurements to the problem. In the following, we present the resulting algorithm. 

\subsection{The KS-ADMM Solver}
\label{sec:x_linear}
The proposed KS-ADMM solver is described in Algorithm~\ref{alg:KS-ADMM}. To extend the batch ADMM to KS-ADMM, we first define an artificial measurement noise $\mathbf{\Sigma}_t = \mathbf{I}/\rho$ and a pseudo-measurement $\mathbf{z}_t = \mathbf{w}_t + \bm{\eta}_t/\rho$, and then rewrite \eqref{eq:x_linear} as
\begin{equation}\label{eq:x_kf_linear}
\begin{split}
&\min_{\mathbf{x}_{1:T}}  \frac{1}{2}  \sum_{t=1}^{T}  \| \mathbf{y}_t -\mathbf{H}_t\mathbf{x}_{t} \|_{\mathbf{R}_t^{-1}}^2 
+  \frac{1}{2} \sum_{t=1}^{T}  \left\|\mathbf{z}_t -  \right. 
\left.  \mathbf{\Omega}_t\mathbf{x}_t \right\|_{ \mathbf{\Sigma}_t^{-1} }^2 \\
& \qquad +  \frac{1}{2}\sum_{t=2}^{{T}} \|\mathbf{x}_{t}-\mathbf{A}_t\mathbf{x}_{t-1} \|_{ \mathbf{Q}_t^{-1} }^2  
         +   \frac{1}{2}\| \mathbf{x}_1 -  \mathbf{m}_1 \|_{ \mathbf{P}_1^{-1} }^2.
\end{split}
\end{equation}
The solution to \eqref{eq:x_kf_linear} can then be computed by running KS on the state estimation problem
\begin{subequations} \label{eq:linsubprob}
\begin{align}
  \mathbf{x}_t &= \mathbf{A}_t\mathbf{x}_{t-1} + \mathbf{q}_t, \\
  \mathbf{y}_t &= \mathbf{H}_t \mathbf{x}_t + \mathbf{r}_t, \\
  \mathbf{z}_t &= \mathbf{\Omega}_t \mathbf{x}_t + \bm{\sigma}_t.
\end{align}
\end{subequations}
Here, $\bm{\sigma}_t$ is an independent random variable with covariance $\bm{\Sigma}_t$. The KS-based solution can be described as a four stage recursive process: prediction, $\mathbf{y}_t$-update, $\mathbf{z}_t$-update, and a RTS smoother which should be performed for $t=1,\ldots,T$. First, the prediction step is given by 
\begin{subequations}\label{eq:linear_pre}
\begin{align}
\label{eq:m}
  \mathbf{m}_t^- &=  \mathbf{A}_t  \mathbf{m}_{t-1},  \\
  \label{eq:p}
  \mathbf{P}_t^- &= \mathbf{A}_t \,\mathbf{P}_{t-1}\,\mathbf{A}_t^\T + \mathbf{Q}_t, 
\end{align}
\end{subequations}
where $\mathbf{m}_t^-$ and $\mathbf{P}_t^-$ are the predicted mean and covariance at the time $t$. Secondly, the update steps for $\mathbf{y}_t$ are given by
\begin{subequations} \label{eq:linear_update_y}
\begin{align}
\label{eq:y_s}
 & \mathbf{S}^y_t =  \mathbf{H}_t \, \mathbf{P}_t^-  \, \mathbf{H}_t^\T  + \mathbf{R}_t, \\
 \label{eq:y_k}
 & \mathbf{K}^y_t = \mathbf{P}_t^- \,  \mathbf{H}_t^\T  \, [\mathbf{S}^y_t]^{-1},\\
 \label{eq:y_m}
 & \mathbf{m}^y_t = \mathbf{m}_t^-  + \mathbf{K}^y_t [\mathbf{y}_t - \mathbf{H}_t \, \mathbf{m}_t^-], \\
  \label{eq:y_p}
 & \mathbf{P}^y_t = \mathbf{P}_t^- - \mathbf{K}^y_t \, \mathbf{S}^y_t \,[\mathbf{K}^y_t]^\T.
\end{align}
\end{subequations}
Thirdly, the update steps for $\mathbf{z}_t$ are
\begin{subequations}  \label{eq:linear_update_z}
\begin{align}
\label{eq:z_s}
  \mathbf{S}^z_t &= \mathbf{\Omega}_t\,\mathbf{P}^y_t \,\mathbf{\Omega}_t^\T + \mathbf{\Sigma_t}, \\
  \label{eq:z_k}
  \mathbf{K}^z_t &= \mathbf{P}^y_t \,\mathbf{\Omega}_t^\T  [\mathbf{S}^z_t]^{-1}, \\
  \label{eq:z_m}
  \mathbf{m}_t &= \mathbf{m}^y_t + \mathbf{K}^z_t[\mathbf{z}_t - \mathbf{\Omega_t}\mathbf{m}^y_t], \\
  \label{eq:z_p}
  \mathbf{P}_t &= \mathbf{P}^y_t - \mathbf{K}^z_t \, \mathbf{S}^z_t \, [\mathbf{K}^z_t]^\T .
\end{align}
\end{subequations}
Here, $\mathbf{S}^y_t$ and $\mathbf{S}^z_t$, $\mathbf{K}^y_t$ and $\mathbf{K}^z_t$, $\mathbf{m}^y_t$ and $\mathbf{m}_t$, $\mathbf{P}^y_t$ and $\mathbf{P}_t$ are the innovation covariances, gain matrices, means, and covariances for the variables $\mathbf{y}_t$ and $\mathbf{z}_t$ at the time step $t$, respectively. Finally, we run a RTS smoother \cite{RTS1965} for $t$ = $T-1, \ldots,1$, which has the steps
\begin{subequations} \label{eq:linear_smoother}
\begin{align}
\label{eq:s_g}
  \mathbf{G}_t   &=   \mathbf{P}_t \, \mathbf{A}_{t+1}^\T  \, [\mathbf{P}^-_{t+1}]^{-1}, \\
  \label{eq:s_m}
  \mathbf{m}_t^s &=   \mathbf{P}_t + \mathbf{G}_t  \, [\mathbf{m}^s_{t+1} - \mathbf{m}^-_{t+1}] ,\\
  \label{eq:s_p}
  \mathbf{P}_t^s &=   \mathbf{P}_t + \mathbf{G}_t \, [\mathbf{P}^s_{t+1} - \mathbf{P}^-_{t+1}] \, \mathbf{G}_t^\T,
\end{align}
\end{subequations}
where $\mathbf{m}_T^s = \mathbf{m}_T$ and $\mathbf{P}_T^s = \mathbf{P}_T$  (see \cite{simo2013Bayesian} for more details). This gives the update for $\mathbf{x}_{1:T}$ as:  
\begin{equation}\label{eq:x_linear_solu}
\begin{split}
\begin{aligned}
\mathbf{x}_{1:T}^{(k+1)} = \mathbf{m}^s_{1:{T}}. 
\end{aligned}
\end{split}
\end{equation}
The remaining updates for $t=1,\ldots,T$ are given as
\begin{equation}
\begin{split}
\begin{aligned}
\label{eq:wt_compute}
&\mathbf{w}_{t}^{(k+1)} = 
\operatorname{max} (| \mathbf{e}_t^{(k)}|- \lambda / \rho,0) \, \operatorname{sgn}(\mathbf{e}_t^{(k)} ) ,
\end{aligned}
\end{split}
\end{equation}
where $\mathbf{e}_t^{(k)}=\mathbf{\Omega}_t \mathbf{x}_t^{(k+1)} + \bm{\eta}_t^{(k)}/\rho$, and
\begin{equation}
\begin{split}
    \bm{\eta}_t^{(k+1)}&= \bm{\eta}_t^{(k)}+ \rho\,(\mathbf{w}_t^{(k+1)} - \mathbf{\Omega}_t\mathbf{x}_t^{(k+1)} ).
\end{split}
\label{eq:et_compute}
\end{equation}
 \begin{algorithm} 
     \caption{KS-ADMM }    \label{alg:KS-ADMM}
      \KwIn{$\mathbf{y}_t$, $\mathbf{H}_t$, $\mathbf{A}_t$, $\mathbf{Q}_t$, $\mathbf{R}_t$, $\mathbf{\Omega}_t$, $t=1,\ldots,T$; parameters $\lambda$ and $\rho$; $\mathbf{m}_1$ and $\mathbf{P}_1$.} 
       \KwOut{$\mathbf{x}_{1:{T}}.$} 
       \While{not convergent} 
        { 
         run the Kalman filter using \eqref{eq:linear_pre}, \eqref{eq:linear_update_y}, and \eqref{eq:linear_update_z}\;
         run the RTS smoother by using \eqref{eq:linear_smoother}\;
         compute $\mathbf{x}_{1:T}$ by \eqref{eq:x_linear_solu}\;  
         compute $\mathbf{w}_{1:T}$ by \eqref{eq:wt_compute}\;
         compute $\bm{\eta}_{1:T}$ by \eqref{eq:et_compute}\;
         }
    \end{algorithm}

It is useful to note that in Algorithm~\ref{alg:KS-ADMM}, the covariances and gains are independent of the iteration number and thus can be pre-computed outside the ADMM iterations. Furthermore, when the model is time-independent, we can often use stationary Kalman filters and smoothers instead of their general counterparts which can further be used to speed up the computations.

\subsection{Convergence of KS-ADMM}
\label{sec:convergence_ksadmm}
In this section, we discuss the convergence of KS-ADMM. If the system \eqref{eq:linsubprob} is detectable \cite{Anderson:1981}, then the objective function~\eqref{eq:Lagrangian_func} is convex. The traditional convergence results for ADMM such as in \cite{Boyd2011admm,He2000parameter} then ensure that the objective globally converges to the stationary (optimal) point $(\mathbf{x}_{1:T}^\star, \mathbf{w}_{1:T}^\star, \bm{\eta}_{1:T}^\star)$. The result is given in the following. 

\begin{theorem}[Convergence of KS-ADMM]
\label{theorem_admm}
Assume that the system \eqref{eq:linsubprob} is detectable \cite{Anderson:1981}. Then, 
for a constant $\rho$, the sequence $\{\mathbf{x}_{1:T}^{(k)}, \mathbf{w}_{1:T}^{(k)}, \bm{\eta}_{1:T}^{(k)}\}$ generated by Algorithm \ref{alg:KS-ADMM} from any starting point $\{\mathbf{x}_{1:T}^{(0)}, \mathbf{w}_{1:T}^{(0)}, \bm{\eta}_{1:T}^{(0)}\}$ converges to the stationary point $\{\mathbf{x}_{1:T}^\star, \mathbf{w}_{1:T}^\star, \bm{\eta}_{1:T}^\star\}$ of \eqref{eq:Lagrangian_func}.
\end{theorem}
\begin{proof}
Due to the detectability assumption, the objective function is convex, and thus the result follows from the classical ADMM convergence proof \cite{Boyd2011admm,He2000parameter}. 
\end{proof}

\section{Nonlinear State Estimation by IEKS-ADMM }
\label{sec:IEKS-ADMM}
When $\mathbf{a}_t$ and $\mathbf{h}_t$ are nonlinear, the $\mathbf{x}$ subproblem arising in the ADMM iteration cannot be solved in closed form. In the following, we first present a batch solution of the nonlinear case based on a Gauss--Newton (GN) iteration and then show how it can be efficiently implemented by IEKS.  

\subsection{Batch Optimization}
\label{sec:Optimization_nonlinear}
Let us now consider the case where the state transition function $\mathbf{a}_t$ and the measurement function $\mathbf{h}_t$ in \eqref{eq:model} are nonlinear. We now proceed to rewrite the optimization~\eqref{eq:optimization} in batch form by defining the following variables
\begin{subequations}
\label{eq:nonlinear_vector_sets}
\begin{align}
&\mathbf{a}(\mathbf{x}) =  \operatorname{vec}(\mathbf{x}_1,\mathbf{x}_2- \mathbf{a}_2(\mathbf{x}_1),
\ldots,\mathbf{x}_T - \mathbf{a}_{T}(\mathbf{x}_{T-1})),\\
&\mathbf{h}(\mathbf{x}) = \operatorname{vec}(\mathbf{h}_1(\mathbf{x}_1),\mathbf{h}_2(\mathbf{x}_2),\ldots,\mathbf{h}_{T}(\mathbf{x}_{T})).
\end{align}
\end{subequations}
Note that the variables $\mathbf{x}$, $\mathbf{y}$, $\mathbf{m}$, $\mathbf{Q}$, $\mathbf{R}$ and $\mathbf{\Omega}$ have the same definitions as \eqref{eq:linear_vector_sets}. Using these variables, the $\mathbf{x}$ subproblem can be naturally transformed into
\begin{equation}\label{eq:nonlinear_whole_optimization} 
\begin{split}
\begin{aligned}
\mathbf{x}^\star  &= \mathop{\arg\min}_{\mathbf{x}} 
\frac{1}{2} \left \| \mathbf{y} - \mathbf{h}(\mathbf{x})\right \|_{\mathbf{R}^{-1}}^2 
+ \frac{1}{2} \left \| \mathbf{m} - \mathbf{a}(\mathbf{x}) \right \|_{\mathbf{Q}^{-1}}^2  \\
& \quad +  \lambda \left \|  \mathbf{\Omega} \mathbf{x} \right \|_1,
\end{aligned}
\end{split}
\end{equation}
which is also a special case of \eqref{eq:general_function}, similarly to the linear case. 

Following the ADMM, we define the augmented Lagrangian function associated with \eqref{eq:nonlinear_whole_optimization} as:
\begin{equation}\label{eq:Lagrangian_func_nonlinear}
\begin{split}
\begin{aligned}
&\mathcal{L}(\mathbf{x},\mathbf{w};\bm{\eta}) \triangleq  
\frac{1}{2} \left \| \mathbf{y} - \mathbf{h}(\mathbf{x})\right \|_{\mathbf{R}^{-1}}^2 
+ \lambda \| \mathbf{w}  \|_1 + \\ 
&  \, \frac{1}{2} \left \| \mathbf{m} - \mathbf{a}(\mathbf{x}) \right \|_{\mathbf{Q}^{-1}}^2  
 + \bm{\eta}^\T(\mathbf{w}- \mathbf{\Omega}\mathbf{x})
+ \frac{\rho}{2} \| \mathbf{w}- \mathbf{\Omega}\mathbf{x}\|^2. 
\end{aligned}
\end{split}
\end{equation}
Since the nonlinear batch solution is based on ADMM, the iteration steps of $\mathbf{w}$ and $\bm{\eta}$ are the same with the linear case (see Equations \eqref{eq:w_compute_2} and \eqref{eq:admm_iteration_eta}). Here, we focus on introducing the solution of the primal variable $\mathbf{x}$. 

When updating $\mathbf{x}$, the objective is no longer a quadratic function. However, the optimization problem can be solved with GN~\cite{Wright1999Optimization}. Here, the $\mathbf{x}$ subproblem is rewritten as 
\begin{equation}
\label{eq:x}
\begin{split}
\begin{aligned}
\underset{\mathbf{x}}{\mathrm{min}}\, f(\mathbf{x}),
\end{aligned}
\end{split}
\end{equation}
where
\begin{equation}
\begin{split}
f(\mathbf{x}) &=
  \frac{1}{2} \| \mathbf{R}^{-\frac{1}{2}}
 (\mathbf{y} - \mathbf{h}(\mathbf{x})) \| ^2 \\
 & \quad  +\frac{1}{2} \| \mathbf{Q}^{-\frac{1}{2}}
  \, (\mathbf{m} - \mathbf{a}(\mathbf{x}))\| ^2
 + \frac{\rho}{2} \| \mathbf{w} - \mathbf{\Omega} \mathbf{x} + \bm{\eta}/\rho\| ^2. \nonumber
\end{split}
\end{equation}
Then, the gradient of $f(\mathbf{x})$ is given by
\begin{equation} \label{eq:gradient_f}
\begin{split}
\begin{aligned}
\nabla  f(\mathbf{x}) 
&= \begin{bmatrix} \mathbf{R}^{-\frac{1}{2}}  \mathbf{J}_{h}(\mathbf{x} ) \\
\mathbf{Q}^{-\frac{1}{2}}  \mathbf{J}_{a}(\mathbf{x} ) \\
\rho^{\frac{1}{2}}  \mathbf{\Omega}
\end{bmatrix} ^\T
\begin{bmatrix}
   \mathbf{R}^{-\frac{1}{2}} (\mathbf{h}(\mathbf{x})- \mathbf{y} ) \\
  \mathbf{Q}^{-\frac{1}{2}}(\mathbf{a}(\mathbf{x}) - \mathbf{m}) \\
  \rho^{\frac{1}{2}} (\mathbf{\Omega} \mathbf{x} - \mathbf{w} - \bm{\eta}/\rho )
\end{bmatrix},
\end{aligned}
\end{split}
\end{equation}
where 
\begin{subequations}
\begin{align}
\mathbf{J}_h(\mathbf{x} ) &= \operatorname{blkdiag}(\mathbf{J}_{h_1},\mathbf{J}_{h_2},\ldots,\mathbf{J}_{h_T}),  \nonumber \\ 
\mathbf{J}_a (\mathbf{x} )&= \begin{pmatrix}
    \mathbf{I}   & \mathbf{0} & & \\
    - \mathbf{J}_{a_2} &  \mathbf{I} &  \ddots & \\
   & \ddots & \ddots & \mathbf{0} \\
   & &  - \mathbf{J}_{a_T} &  \mathbf{I} \\
  \end{pmatrix},  \nonumber
\end{align}
\end{subequations}
and the Hessian is $\nabla^2  f(\mathbf{x}) = \mathbf{J}^\top  \mathbf{J}(\mathbf{x}) + \mathbf{H}(\mathbf{x})$, 
where 
\begin{equation}
\label{eq:gradient}
\begin{split}
\begin{aligned}
\mathbf{J}^\top  \mathbf{J}(\mathbf{x}) &= \mathbf{J}_{h}^\T \mathbf{R}^{-1} \mathbf{J}_{h}(\mathbf{x} )+ 
     \mathbf{J}_{a}^\T \mathbf{Q}^{-1} \mathbf{J}_{a}(\mathbf{x} ) +  \rho \mathbf{\Omega}^\top \mathbf{\Omega}, \\
[\mathbf{H}(\mathbf{x})]_{ij} &= 
  \frac{1}{2} (\mathbf{h}(\mathbf{x}) - \mathbf{y})^\top \mathbf{R}^{-1} 
 \frac{\partial^2 \mathbf{h}(\mathbf{x})}{\partial \mathbf{x}_i \, \partial \mathbf{x}_j}  \nonumber \\
& \quad + \frac{1}{2} (\mathbf{a}(\mathbf{x}) - \mathbf{m})^\top \mathbf{Q}^{-1}
  \, \frac{\partial^2 \mathbf{a}(\mathbf{x})}{\partial \mathbf{x}_i \, \partial \mathbf{x}_j}.\nonumber
\end{aligned}
\end{split}
\end{equation}

In GN, avoiding the trouble of computing the residual $[\mathbf{H}(\mathbf{x})]_{ij}$, we use the approximation $\nabla^2  f(\mathbf{x}) \approx \mathbf{J}^\top  \mathbf{J}(\mathbf{x})$ to replace $\nabla^2 f(\mathbf{x}) $, which means $[\mathbf{H}(\mathbf{x})]_{ij}$ is assumed to be small enough.  
Thus, the primal variable in $\mathbf{x}$ iteration is updated by:
\begin{equation}
\label{eq:x_nonlinear_iteration}
\begin{split}
\mathbf{x}^\star &= \left[ \mathbf{J}_h^\T \mathbf{R}^{-1} \mathbf{J}_h(\mathbf{x}) 
   + \mathbf{J}_a^\T \mathbf{Q}^{-1}\mathbf{J}_a(\mathbf{x})
   + \rho \, \mathbf{\Omega}^\T\mathbf{\Omega} \right]^{-1} \\
   & \quad  \times [ \mathbf{J}_h^\T \mathbf{R}^{-1} ( \mathbf{y}  - \mathbf{h}(\mathbf{x})) + 
  \mathbf{J}_a^\T \mathbf{Q}^{-1} ( \mathbf{m} - \mathbf{a}(\mathbf{x})) \\
   & \quad + \rho \, \mathbf{\Omega}^\T ( \mathbf{w}^{(k)} + \bm{\eta}^{(k)}/\rho )  ].
\end{split}
\end{equation}
The iterations can stop after a maximum number of iterations $i_{\max}$ or if the condition $\|\mathbf{x}^{(i+1)} - \mathbf{x}^{(i)} \|_2 \leq \varepsilon$
is satisfied, where $\varepsilon$ is an error tolerance. If $\varepsilon$ is small enough, then it means that the above algorithm has (almost) converged.  The rest of the ADMM updates can be implemented similarly to the linear Gaussian case.

\subsection{The IEKS-ADMM Solver}
\label{section:x_nonlinear}
We now move on to consider the IEKS-ADMM solver. As discussed in Section~\ref{sec:ieks}, IEKS can be seen as an efficient implementation of the GN method, which inspires us to derive an efficient implementation of the batch ADMM.

Now, we rewrite the $\mathbf{x}$ subproblem \eqref{eq:x} as 
\begin{equation}
\label{eq:x_nonlinear_iekf}
\begin{split}
&\min_{\mathbf{x}_{1:T}} 
\frac{1}{2}  \sum_{t=1}^{T}  \| \mathbf{y}_t -\mathbf{h}_t(\mathbf{x}_t)\|_{\mathbf{R}_t^{-1}}^2 
+  \frac{1}{2} \sum_{t=1}^{T}  \left\|\mathbf{z}_t - \mathbf{\Omega}_t\mathbf{x}_t \right\|_{ \mathbf{\Sigma}_t^{-1} }^2 \\
& + \frac{1}{2}\sum_{t=2}^{{T}} \|\mathbf{x}_{t}-\mathbf{a}_t(\mathbf{x}_{t-1})\|_{ \mathbf{Q}_t^{-1} }^2   
+   \frac{1}{2}\| \mathbf{x}_1 -  \mathbf{m}_1 \|_{ \mathbf{P}_1^{-1} }^2.  
\end{split}
\end{equation}
In a modest scale (e.g., $T \approx 10^3$), $\mathbf{x}_{1:T}$ can be directly computed by \eqref{eq:x_nonlinear_iteration} although its computations scale as $\mathcal{O}(n_x^3 \times T^3)$. When $T$ is large, the batch ADMM will have high memory and computational requirements. In this case, the use of IEKS becomes beneficial due to its linear computational scaling. In this paper, the proposed method incorporates IEKS into ADMM to design the \mbox{IEKS-ADMM} algorithm for solving the nonlinear case. 

In the IEKS algorithm, the Gaussian smoother is run several times with $\mathbf{a}_t$ and $\mathbf{h}_t$ and their Jacobians are evaluated at the previous (inner loop) iteration. The detailed iteration steps of \mbox{IEKS-ADMM} are described in Algorithm~\ref{alg:IEKS-ADMM}.  In particular, following the prediction steps \eqref{eq:nonliear_pre} in Section \ref{sec:ieks}, the update steps for $\mathbf{y}_t$ are given by
\begin{subequations} \label{eq:nonliear_update_y}
\begin{align}
  &\mathbf{S}^y_t = \mathbf{J}_{h_t}(\mathbf{x}_t^{(i)}) \, \mathbf{P}_t^-  \,[\mathbf{J}_{h_t}(\mathbf{x}_t^{(i)})]^\T + \mathbf{R}_t, \\
  &\mathbf{K}^y_t = \mathbf{P}_t^- \,[\mathbf{J}_{h_t}(\mathbf{x}_t^{(i)})]^\T  \, [\mathbf{S}^y_t]^{-1}, \\
 & \mathbf{m}^y_t = \mathbf{m}_t^-  + \mathbf{K}^y_t [\mathbf{y}_t - \mathbf{h}_t(\mathbf{x}_t^{(i)}) - \mathbf{J}_{h_t}(\mathbf{x}_t^{(i)}) (\mathbf{m}_t^- - \mathbf{x}_t^{(i)})], \\
& \mathbf{P}^y_t = \mathbf{P}_t^- - \mathbf{K}^y_t \, \mathbf{S}^y_t \,[\mathbf{K}^y_t]^\T, 
\end{align}
\end{subequations}
and for the pseudo-measurement $\mathbf{z}_t$, the update steps are the same as in the linear case. They are given in~\eqref{eq:linear_update_z}.

Additionally, the RTS smoother steps are also described in Section \ref{sec:ieks}. We can then obtain the solution as $\mathbf{x}_{1:T}^{(k+1)} = \mathbf{m}^s_{1:{T}}$.
Note that the updates on $\mathbf{w}_{1:T}$ and  $\bm{\eta}_{1:T}$ can be implemented by \eqref{eq:wt_compute} and \eqref{eq:et_compute}, respectively. 
 
 \begin{algorithm} 
     \caption{IEKS-ADMM }    \label{alg:IEKS-ADMM}
      \KwIn{$\mathbf{y}_t$, $\mathbf{h}_t$, $\mathbf{a}_t$, $\mathbf{Q}_t$, $\mathbf{R}_t$, $\mathbf{\Omega}_t$, $t=1,\ldots,T$; parameters $\lambda$ and $\rho$;
      $\mathbf{m}_1$ and $\mathbf{P}_1$. } 
       \KwOut{$\mathbf{x}_{1:{T}}$} 
       \While{not convergent} 
        { 
         compute $\mathbf{x}_{1:T}$ by using the IEKS\;
         compute $\mathbf{w}_{1:T}$ by \eqref{eq:wt_compute}\;
         compute $\bm{\eta}_{1:T}$ by \eqref{eq:et_compute}\;
         }
    \end{algorithm}

\subsection{Convergence of IEKS-ADMM} 
\label{sec:ieks_convergence}
In this section, our aim is to prove the convergence of the IEKS-ADMM algorithm. Although we can
rely much on existing convergence results, unfortunately, when $\mathbf{a}(\mathbf{x})$ and $\mathbf{h}(\mathbf{x})$ are nonlinear, the traditional convergence analysis \cite{Boyd2011admm,Hong2016Convergence,Deng2016convergence,admm2019convergence} does not work as such. In particular, Jacobian matrices $\mathbf{J}_{a}$, $\mathbf{J}_{h}$ and linear operator $\mathbf{\Omega}$ in this paper are possibly rank-deficient, which is not covered by the existing convergence results. In the following, we will establish the convergence analysis which also covers this case. 

For notational convenience, we define 
$\theta_1(\mathbf{x}) =
\frac{1}{2}\| \mathbf{y} -\mathbf{h}(\mathbf{x})  \|_{\mathbf{R}^{-1}}^2
+ \frac{1}{2} \left \| \mathbf{m} - \mathbf{a}(\mathbf{x}) \right \|_{\mathbf{Q}^{-1}}^2 $ and 
$\theta_2(\mathbf{w}) = \lambda  \| \mathbf{w}  \|_1$. 
The variables $\mathbf{x} $ and $\mathbf{w}$ are two sets of time series, and $\theta_1(\mathbf{x})$ is a non-quadratic, possibly nonconvex function. The corresponding augmented Lagrangian function can be rewritten as

\begin{equation}\label{eq:Lagrangian_proof}
\begin{split}
\begin{aligned}
\mathcal{L}(\mathbf{x},\mathbf{w};\bm{\eta})  & = 
\theta_1(\mathbf{x}) + \theta_2(\mathbf{w})  \\
& \quad + \bm{\eta}^\T(\mathbf{w}- \mathbf{\Omega}\mathbf{x})
+ \frac{\rho}{2} \| \mathbf{w}- \mathbf{\Omega}\mathbf{x}\|^2,
\end{aligned}
\end{split}
\end{equation}
where $\mathbf{\Omega}$ can be full-row rank or full-column rank. Thus, the convergence is analyzed in two different cases. We make the following assumptions.
\begin{assumption}
	\label{assump:basic_Lipschitz}
The gradient $\nabla \theta_1(\mathbf{x})$ is Lipschitz continuous with constant $L_{\theta_1}$, that is, 
\begin{equation}
\begin{split}
 \| \nabla \theta_1(\mathbf{x}_1)  - \nabla \theta_1(\mathbf{x}_2) \|  \leq {L_{\theta_1} }  \| \mathbf{x}_1 - \mathbf{x}_2 \|,  \forall \mathbf{x}_1, \mathbf{x}_2 \in \text{dom}( \theta_1).
 \nonumber
\end{split}
\end{equation}
\end{assumption}
\begin{assumption}
	\label{assump:basic_coercive}
Function $\theta_1(\mathbf{x}) + \theta_2(\mathbf{w}) $ is lower bounded and coercive over the feasible set 
$\{ (\mathbf{x},\mathbf{w}) : \mathbf{w} = \mathbf{\Omega} \mathbf{x}\}$.
\end{assumption}
First, we prove that the sequence $\mathcal{L}(\mathbf{x}^{(k)},\mathbf{w}^{(k)};\bm{\eta}^{(k)}) $ is monotonically non-increasing in the following lemma.

\begin{lemma}[Nonincreasing sequence]
\label{lemma:nonincreasing}

Let Assumptions \ref{assump:basic_Lipschitz} and \ref{assump:basic_coercive} be satisfied and $\{\mathbf{x}^{(k)},\mathbf{w}^{(k)}, \bm{\eta}^{(k)}\}$ be the iterative sequence generated by ADMM. Assume that one of two cases is satisfied: 

\textbf{Case (a):}  There exists $\rho_0$ such that when $\rho > \rho_0$, $\mathbf{x}  \mapsto  \mathcal{L}(\mathbf{x},\mathbf{w};\bm{\eta})$  is $\mu_x$-strongly convex, that is, $\mathcal{L}(\mathbf{x},\mathbf{w};\bm{\eta})$ satisfies
\begin{equation}\label{eq:Lagrangian_func_short}
\begin{split}
&  \mathcal{L}(\mathbf{x}^{(k)},\mathbf{w}^{(k)};\bm{\eta}^{(k)})  -   \mathcal{L}(\mathbf{x}^{(k+1)},\mathbf{w}^{(k)};\bm{\eta}^{(k)}) \\
& \geq 
\langle \nabla \mathcal{L}(\mathbf{x}^{(k)},\mathbf{w}^{(k)};\bm{\eta}^{(k)}), \mathbf{x}^{(k)} -\mathbf{x}^{(k+1)} \rangle  \\
& \quad+ \frac{\mu_x}{2} \| \mathbf{x}^{(k)} -\mathbf{x}^{(k+1)}  \|^2.
\end{split}
\end{equation}
Furthermore, assume that $\rho > \max\left( \frac{2 L_{\theta_1}^2}{\kappa_a^2 \mu_x},\rho_0\right)$ and that $\mathbf{\Omega}$ has full row rank with 
\begin{equation}
\begin{split}
\mathbf{\Omega}\,\mathbf{\Omega}^\T \succeq \kappa_a^2 \mathbf{I},  \quad  \kappa_a> 0. 
\end{split}
\end{equation}

\textbf{Case (b):} $\rho > \frac{L_{\theta_1} }{  \kappa_b^2 }$, and $\mathbf{\Omega}$ has full-column rank with 
\begin{equation}
\begin{split}
\mathbf{\Omega}^\T \mathbf{\Omega} \succeq \kappa_b^2 \mathbf{I}, \quad \kappa_b > 0. 
\end{split}
\end{equation}
Then, sequence $\mathcal{L}(\mathbf{x}^{(k)},\mathbf{w}^{(k)};\bm{\eta}^{(k)}) $ is nonincreasing in $k$. 
\end{lemma}

\begin{proof}
See Appendix \ref{pf:lemma:nonincreasing}. 
\end{proof}
Next we prove the convergence of Algorithm \ref{alg:IEKS-ADMM}. For that we need a couple of lemmas which are presented in the following.

\begin{lemma} [Convergence of GN]
\label{lemma:gn}
Let $ \nabla f(\mathbf{x})$ be Lipschitz continuous with constant $L_{f}$, $\mathbf{H}(\mathbf{x})$ be bounded by a constant, that is, $\|\mathbf{H}(\mathbf{x}) \| \le \epsilon_h $. If $\mathbf{J}^\T \mathbf{J}(\mathbf{x}^{(i)})  \succeq \mu^2 \mathbf{I}$ where $\mu > 0$ is a constant, and $ \epsilon_h  < \mu^2$, 
then the sequence $\mathbf{x}^{(i)}$ converges to a local minimum $\mathbf{x}^\star$. In particular, the convergence is quadratic when $\epsilon_h \to 0$, and (at least) linear convergence is obtained when $\epsilon_h < \mu^2$.  
\end{lemma}

\begin{proof}
See Appendix \ref{pf:lemma_gn}.	  
\end{proof}

Based on Lemmas \ref{lemma:nonincreasing} and \ref{lemma:gn}, we can now establish the convergence by GN-ADMM in the following lemma. 

\begin{lemma}[Convergence of GN-ADMM]
\label{lemma:GN-ADMM} 
Let assumptions of Lemmas~\ref{lemma:nonincreasing} and \ref{lemma:gn} be satisfied.
Then, the sequence $\{\mathbf{x}^{(k)}, \mathbf{w}^{(k)}, \bm{\eta}^{(k)}\}$ generated by GN-ADMM algorithm converges to a local minimum $(\mathbf{x}^\star, \mathbf{w}^\star, \bm{\eta}^\star).$ 
\end{lemma}

\begin{proof}
By Lemma~\ref{lemma:nonincreasing}, the sequence $\mathcal{L}(\mathbf{x}^{(k)}, \mathbf{w}^{(k)}; \bm{\eta}^{(k)})$ is nonincreasing in $k$. By Assumption \ref{assump:basic_coercive}, the sequence $\{\mathbf{x}^{(k)}, \mathbf{w}^{(k)}, \bm{\eta}^{(k)}\}$ is bounded, because $\mathcal{L}(\mathbf{x}^{(k)}, \mathbf{w}^{(k)}; \bm{\eta}^{(k)})$ is upper bounded by $\mathcal{L}(\mathbf{x}^{(0)}, \mathbf{w}^{(0)}; \bm{\eta}^{(0)})$ and nonincreasing. It is also lower bounded by 
\begin{equation}
\begin{split}
\mathcal{L}(\mathbf{x}^{(k)}, \mathbf{w}^{(k)}; \bm{\eta}^{(k)}) \geq \theta_1(\mathbf{x}^{(k)}) +\theta_2(\mathbf{w}^{(k)}). 
\end{split}
\end{equation}
By Lemma \ref{lemma:gn}, there exists a local minimum $\mathbf{x}^{\star}$ such that the sequence ${\mathbf{x}^{(i)}}$ converges to $\mathbf{x}^{\star}$, which is a local minimum of $\mathbf{x}$ subproblem. The $\mathbf{w}$ subproblem is convex \cite{boyd2004Convex} and thus there exists a unique minimum $\mathbf{w}^{\star}$. We then deduce that the iterative sequence $\{\mathbf{x}^{(k)}, \mathbf{w}^{(k)}, \bm{\eta}^{(k)}\}$ generated by GN-ADMM converges to $(\mathbf{x}^\star, \mathbf{w}^\star, \bm{\eta}^\star)$. 
\end{proof}

The equivalence of GN and IEKS can now be used to show that IEKS-ADMM converges to a local minimum $(\mathbf{x}^\star_{1:T}, \mathbf{w}^\star_{1:T},\bm{\eta}^\star_{1:T})$.

\begin{theorem}[Convergence of IEKS-ADMM]
\label{theorem1}
If the sequence $\{\mathbf{x}^{(k)}, \mathbf{w}^{(k)}, \bm{\eta}^{(k)}\}$ generated by GN-ADMM algorithm converges to a local minimum $(\mathbf{x}^\star, \mathbf{w}^\star, \bm{\eta}^\star)$, then the sequence $\{\mathbf{x}^{(k)}_{1:T}, \mathbf{w}^{(k)}_{1:T}, \bm{\eta}^{(k)}_{1:T}\}$ generated by IEKS-ADMM algorithm converges to the local minimum $(\mathbf{x}^\star_{1:T}, \mathbf{w}^\star_{1:T}, \bm{\eta}^\star_{1:T}).$ 
\end{theorem}

\begin{proof}
According to\cite{Bell1993ekf}, the sequence $\{\mathbf{x}^{(k)}, \mathbf{w}^{(k)}, \bm{\eta}^{(k)}\}$ generated by the GN method and the sequence $\{\mathbf{x}^{(k)}_{1:T}, \mathbf{w}^{(k)}_{1:T}, \bm{\eta}^{(k)}_{1:T}\}$ generated by IEKS are identical. Based on Lemma \ref{lemma:GN-ADMM}, we deduce the iterative sequence $\{\mathbf{x}^{(k)}_{1:T}, \mathbf{w}^{(k)}_{1:T}, \bm{\eta}^{(k)}_{1:T}\}$ generated by IEKS-ADMM is locally convergent to $(\mathbf{x}^\star_{1:T}, \mathbf{w}^\star_{1:T}, \bm{\eta}^\star_{1:T})$.
\end{proof} 

\section{Extension to general algorithmic Framework}
\label{sec:framework}

\subsection{The Proposed Framework}
In this subsection, we present a general algorithmic framework based on the combination of the extended Kalman smoother and variable splitting. As the smoother solution only applies to the $\mathbf{x}_{1:T}$-subproblem, here we only formulate the corresponding $\mathbf{x}_{1:T}$-subproblem which can be solved with IEKS. The different variants in the proposed framework are distinguished by three different choices: the pseudo-measurement, the pseudo-measurement covariance, and the pseudo-measurement model matrix. 

When $\mathbf{a}_t$ and $\mathbf{h}_t$ are linear functions, we have the following general objective function for the $\mathbf{x}_{1:T}$-subproblem:
\begin{equation}\label{eq:x_kf_linear_general}
\begin{split}
 &\min_{\mathbf{x}_{1:T}}  \frac{1}{2}  \sum_{t=1}^{T}  \| \mathbf{y}_t -\mathbf{H}_t\mathbf{x}_{t} \|_{\mathbf{R}_t^{-1}}^2
+ 
\frac{1}{2}\sum_{t=2}^{{T}} \|\mathbf{x}_{t}-\mathbf{A}_t\mathbf{x}_{t-1} \|_{ \mathbf{Q}_t^{-1} }^2  
\\
& \quad +   \frac{1}{2}\| \mathbf{x}_1 -  \mathbf{m}_1 \|_{ \mathbf{P}_1^{-1} }^2 
+ \frac{1}{2}  \sum_{t=1}^{T}  \left\|\mathbf{\Delta}_{t} - \mathbf{\Theta}_{t}\mathbf{x}_t \right\|_{ \mathbf{\Sigma}_{t}^{-1} }^2, 
\end{split}
\end{equation}
and when $\mathbf{a}_t$ and $\mathbf{h}_t$ are nonlinear functions, we have
\begin{equation}
\label{eq:x_nonlinear_iekf_general}
\begin{split}
&\min_{\mathbf{x}_{1:T}} 
\frac{1}{2}  \sum_{t=1}^{T}  \| \mathbf{y}_t -\mathbf{h}_t(\mathbf{x}_t)\|_{\mathbf{R}_t^{-1}}^2 +
 \frac{1}{2}\sum_{t=2}^{{T}} \|\mathbf{x}_{t}-\mathbf{a}_t(\mathbf{x}_{t-1})\|_{ \mathbf{Q}_t^{-1} }^2 
 \\
& \quad+  \frac{1}{2}\| \mathbf{x}_1 -  \mathbf{m}_1 \|_{ \mathbf{P}_1^{-1} }^2
+ \frac{1}{2}  \sum_{t=1}^{T}  \left\|\mathbf{\Delta}_{t} - \mathbf{\Theta}_{t}\mathbf{x}_t \right\|_{ \mathbf{\Sigma}_{t}^{-1} }^2.
\end{split}
\end{equation}
In the above objective functions, $\mathbf{\Delta}_{t}$ is the pseudo-measurement, $\mathbf{\Sigma}_{t}$ is the pseudo-measurement covariance, and $\mathbf{\Theta}_{t}$ is the pseudo-measurement model matrix. 

\begin{table*}[h]
\caption{Different choices of IEKS-based variable splitting algorithms }                                                                                                                                                                                                                                                                                                                                                                                                                                       
\begin{center}
 \centering
\begin{tabular}{|c|c|c|c|c|}
\hline
 {Method}  & {Related Quadratic Term} & {$\mathbf{\Theta}_{t}$} & {$\mathbf{\Delta}_{t}$} & {$\mathbf{\Sigma}_{t}$}  \\
\hline
PRS  & 
$\frac{\rho}{2} \| \mathbf{w}_t - \mathbf{\Omega}_t\mathbf{x}_t + \bm{\eta}_t/ \rho\|^2$ 
& $\mathbf{\Omega}_t$
& $ \mathbf{w}_t + \bm{\eta}_t/\rho$ & $ \mathbf{I}/{\rho}$ 
\\ 
\hline
SBM  &
$\frac{\rho}{2}\left\| \mathbf{w}_t   - \mathbf{\Omega}_t\mathbf{x}_t  +  \bm{\eta}_t \right\|^2$
& $\mathbf{\Omega}_t$
&$ \bm{\eta}_t + \mathbf{w}_t$ & $ \rho \mathbf{I}$ 
\\ 
\hline
FOPD &
$\frac{1}{2\rho}\|\mathbf{x}_t- ( \mathbf{x}_t^{(k)} - \mathbf{\Omega}_t^\T  \mathbf{w}_t )\|^2$
& $\mathbf{I}$ &
$ \mathbf{x}_t^{(k)} - \mathbf{\Omega}_t^\T  \mathbf{w}_t$ & $ \rho \mathbf{I}$ \\
\hline
ADMM & 
$\frac{\rho}{2} \| \mathbf{w}_t - \mathbf{\Omega}_t\mathbf{x}_t + \bm{\eta}_t/ \rho\|^2$
& $\mathbf{\Omega}_t$
& $ \mathbf{w}_t + \bm{\eta}_t/\rho$ & ${\mathbf{I}}/{\rho}$ 
\\  
\hline 
\end{tabular}
\label{tab1}
\end{center}
\end{table*}

As mentioned in Section \ref{sec:various_methods}, various variable splitting such as PRS, SBM, and FOPD can be used to solve the problems \eqref{eq:x_linear} and \eqref{eq:x}. Their KS / IEKS-based counterparts can be obtained by selecting the aforementioned pseudo-measurement model parameters as shown in Table~\ref{tab1}. The algorithms for solving the optimization problems are then the same as Algorithms~\ref{alg:KS-ADMM} and \ref{alg:IEKS-ADMM} except that the pseudo-measurement updates in~\eqref{eq:linear_update_z} are replaced with
       \begin{subequations}  \label{eq:nonlinear_update_general}
       \begin{align}
       \mathbf{S}^{\delta}_t &= \mathbf{\Theta}_{t} \, \mathbf{P}^y \, \mathbf{\Theta}_{t}^\T + \mathbf{\Sigma}_{t}, \\
       \mathbf{K}^{\delta}_t &= \mathbf{P}_t^- \, \mathbf{\Theta}_{t}^\T \, [\mathbf{S}^{\delta}_t]^{-1}, \\
       \mathbf{m}_t &= \mathbf{m}^y_t + \mathbf{K}^{\delta}_{t} \,
                       [\mathbf{\Delta}_{t} - \mathbf{\Theta}_{t}\,\mathbf{m}^y_t], \\
       \mathbf{P}_t &= \mathbf{P}^y_t - \mathbf{K}^{\delta}_t \, \mathbf{S}^{\delta}_t \, [\mathbf{K}^{\delta}_t]^\T,
        \end{align}
       \end{subequations}
and the updates of the other variables are performed using the appropriate algorithm (see Section~\ref{sec:various_methods}).

\subsection{Computational Complexity}
\label{sec:computation}
This section investigates the computational complexity of the KS / IEKS based variable splitting methods. The proposed methods are iterative, in that we use several numbers of iterations to compute the minimal points also for the primal variable $\mathbf{x}^\star_{1:T}$ in \eqref{eq:nonlinear_whole_optimization}. However, we can always use a bounded number of iterations and thus we only need to determine the complexity of a single iteration to determine the complexity of the whole method. In our case, the computational burden of the auxiliary variable and the dual variable is low compared with the matrix inversions in the primal variable update. Asymptotically, the computational complexities of  \eqref{eq:wt_compute} and \eqref{eq:et_compute} are both $\mathcal{O}(n_x^2 \, T)$. 

In our method, we compute the primal variable update using KS and IEKS, instead of computing matrix inversions explicitly. The time complexity of (iteration of) KS and IEKS is $\mathcal{O}(n_x^3 \, T)$ \cite{Bell1994smoother,simo2013Bayesian,simo2008RTS} (assuming $n_y \le n_x$), while the dominating computation in batch variable splitting methods is the matrix inversion with $\mathcal{O}(n_x^3 \, T^3)$ complexity \cite{Boyd2011admm}. 
Because of the total $\mathcal{O}(n_x^3 \, T)$ computational complexity, the proposed method is especially applicable to large-scale dataset.

\section{Numerical Experiments}
\label{sec:Experiments}
In the following, we demonstrate the KS and IEKS based variable splitting methods in numerical experiments. We first provide several simulated results to study the performance with varying regularization parameter. Then, we turn our attention to the behavior of the proposed methods with respect to the convergence curve and the computational efficiency. Finally, we report the results for large-scale signal estimation and demonstrate the effectiveness of the methodology in a tomographic reconstruction task.

\subsection{Linear Gaussian Simulation Experiment}
\label{results:linear}
Consider a four-dimensional linear tracking model (see, e.g., \cite{simo2013Bayesian}) where the state contains rectangular coordinates $x_1$ and $x_2$, and velocity variables $x_3$ and $x_4$. The state of the system at time step $t$ is $\mathbf{x}_t = \begin{bmatrix} x_{1,t} & x_{2,t} & x_{3,t} & x_{4,t} \end{bmatrix}^\T$. The transition and measurement model matrices are 
$$
\mathbf{A}_t = 
\begin{bmatrix}
1 & 0 & \triangle t &0 \\
0 & 1 & 0 &\triangle t \\
0 & 0 & 1 & 0 \\
0 & 0 & 0 & 1 
\end{bmatrix},
\mathbf{H}_t = 
\begin{bmatrix}
1 & 0 & 0 &0 \\
0 & 1 & 0 & 0
\end{bmatrix}. 
$$
The matrix $\mathbf{\Omega}_t$ and the covariance for the transition are
$$
\mathbf{\Omega}_t = 
\begin{bmatrix}
0 & 0 & 1 &0 \\
0 & 0 & 0 & 1
\end{bmatrix},
\mathbf{Q}_t = q_c \,
\begin{bmatrix} 
\frac{\Delta t^3 }{3} & 0 & \frac{\Delta t^2 }{2} & 0 \\
0 & \frac{\Delta t^3 }{3}  & 0 & \frac{\Delta t^2 }{2}  \\
 \frac{\Delta t^2 }{2}  & 0 &  {\Delta t }  & 0 \\
0 &  \frac{\Delta t^2 }{2}  & 0 & {\Delta t } 
\end{bmatrix}.
$$
with $q_c = 1/2$, $\Delta t = 0.1$, the measurement noise covariance $\mathbf{R}_t = \sigma^2 \mathbf{I}$ with $\sigma = 0.2$, and $T = 100$ (small scale). The relative error is calculated by 
\begin{equation} \label{rho:parameter_function}
\begin{split}
\begin{aligned}
\frac{\sum_{t=1}^{T} \| \mathbf{x}_t^{(k)}-  \mathbf{x}_t^{\text{true}}  \|_2}{\sum_{t=1}^{T} \|\mathbf{x}_t^{\text{true}}\|_2},
\end{aligned}
\end{split}
\end{equation} 
where $\mathbf{x}_t^{\text{true}}$ is the ground truth at time step $t$ and $\mathbf{x}_t^{(k)}$ is the $k$:th iterate at time step $t$. The goal here is to estimate dynamic signals from the noisy measurements $\mathbf{y}_{1:T}$. 

In this experiment, we first illustrate the computations for $1000$ values of the regularizing penalty parameter $\lambda$ in the interval $[0.01, \, 10]$. We remark that other parameters, for example, the parameter $\rho$ in ADMM and $\mbox{KS-ADMM}$, are chosen according to the existing guidelines \cite{Boyd2011admm}, with no aim at further optimizing the convergence performance. 
The CPU times of various KS-based variable splitting methods are listed in Fig.~\ref{fig:time_lam}. The plotted result is an average over $30$ experiments. For $1000$ values of parameter $\lambda$, the total number of ADMM iterations required is less than $20$, which takes around $0.02$ seconds in total. Thus, in the small-scale dataset, the parameter $\lambda$ has a less effect on the computational complexity when $\lambda$ is varying. Fig.~\ref{fig:1d_err_lam} shows the relative error as a function of regularization parameter $\lambda$, using \mbox{KS-PRS}, \mbox{KS-SBM}, \mbox{KS-FOPD}, and \mbox{KS-ADMM}.  As expected, we observe that the relative error is dependent on a proper choice of $\lambda$. We test different methods for the parameter $\lambda$, and find empirically that the lowest relative errors are achieved with $\lambda =1 $. 

\begin{figure}[h]
	\begin{minipage}[t]{0.49\linewidth}
		\centering
		\includegraphics[width=0.98\textwidth, height=0.7\textwidth]{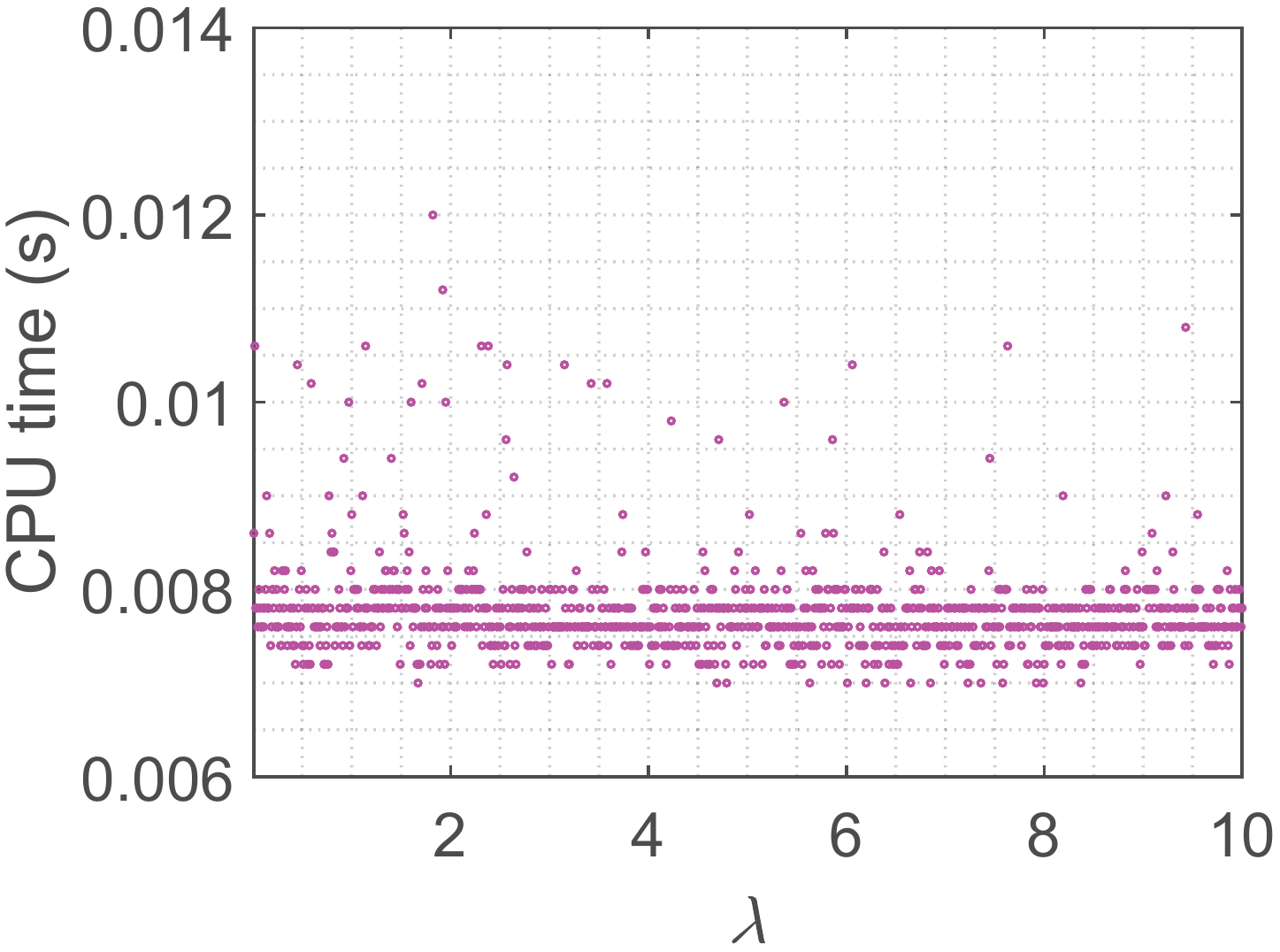}
	\end{minipage}
	\begin{minipage}[t]{0.49\linewidth}
	\centering
	\includegraphics[width=0.98\textwidth, height=0.7\textwidth]{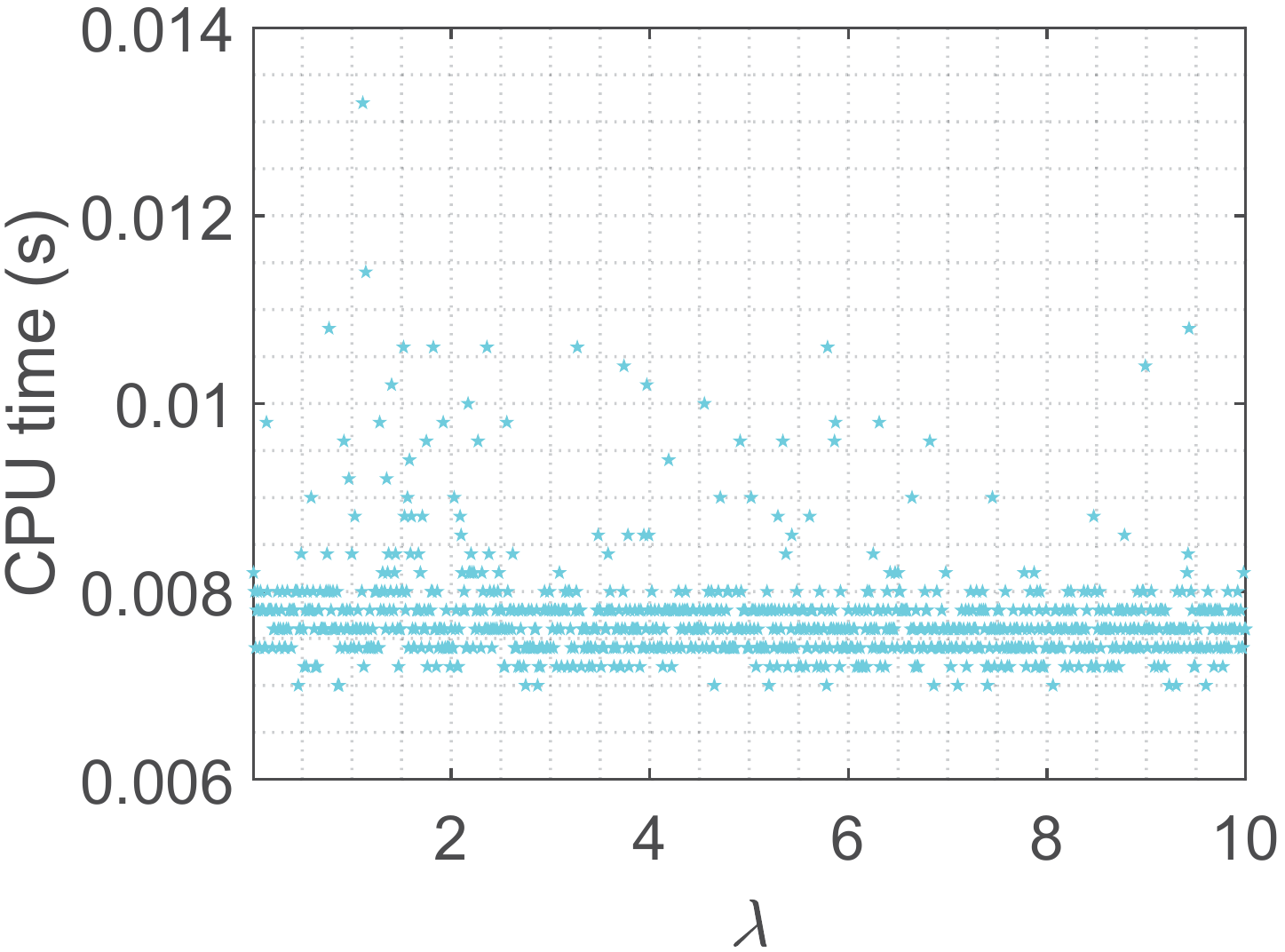}
\end{minipage}
	\\
	\begin{minipage}[t]{0.49\linewidth}
		\centering
		\includegraphics[width=0.98\textwidth, height=0.7\textwidth]{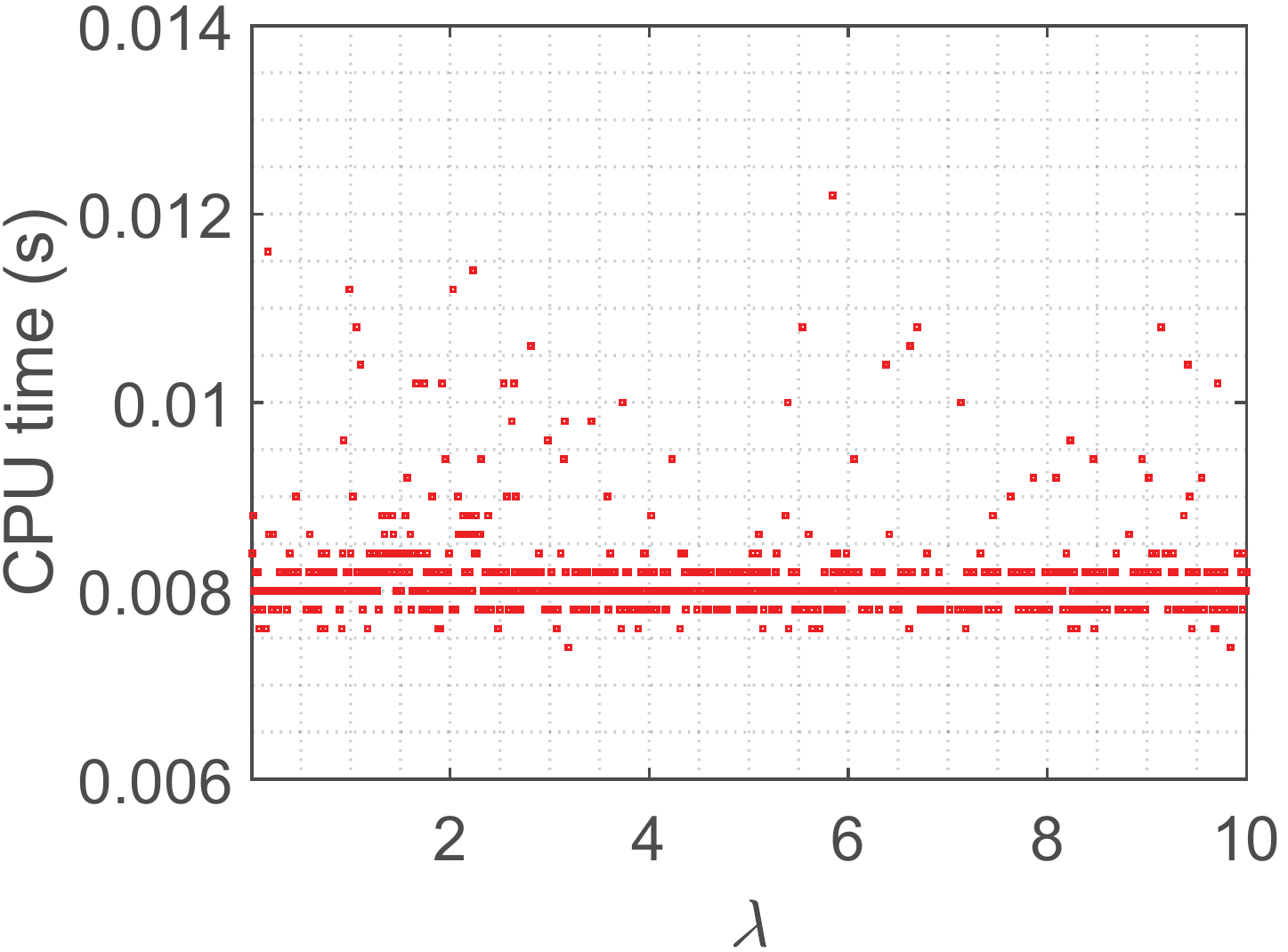}
	\end{minipage}
	\begin{minipage}[t]{0.49\linewidth}
	\centering
	\includegraphics[width=0.98\textwidth, height=0.7\textwidth]{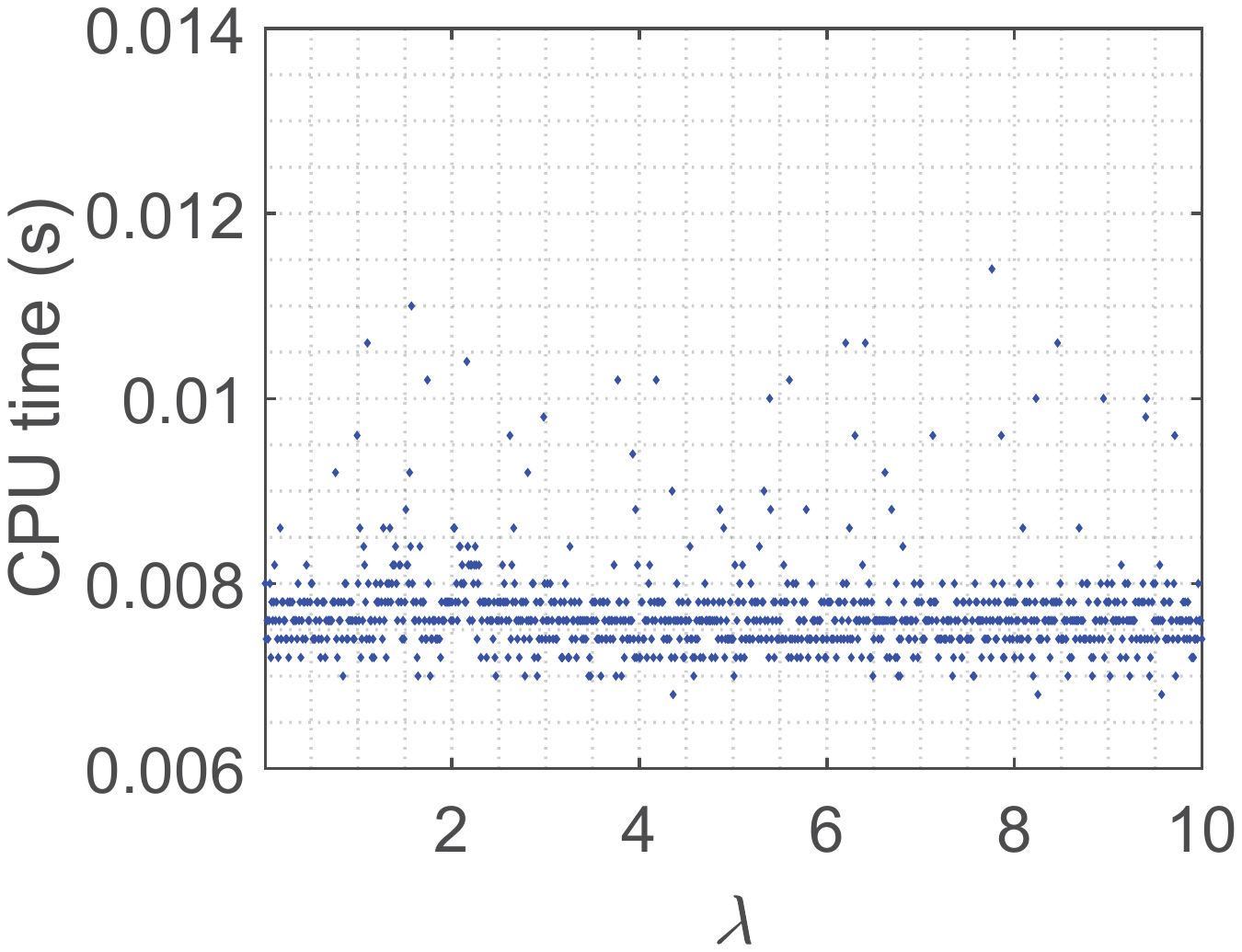}
\end{minipage}
\caption{Average CPU time when increasing the value of $\lambda$ for KS-PRS, KS-SBM, KS-FOPD, and KS-ADMM (from left-top to right-bottom, respectively).}
	\label{fig:time_lam}
\end{figure}
\begin{figure}[!htb]  
	\centerline{\includegraphics[width=0.48\textwidth, height=0.32\textwidth]
		{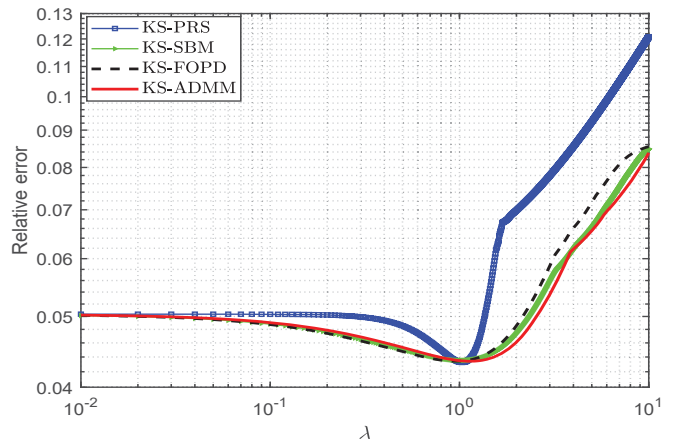}}
	\caption{Relative error as function of the parameter $\lambda$ in the four-dimensional linear tracking model.}
	\label{fig:1d_err_lam}
\end{figure}

Additionally, we compare the convergence speed (relative error versus iteration number) and the running time (average CPU time versus iteration number) generated by batch versions of PRS \cite{He2014peaceman}, SBM~\cite{Goldstein2009Split}, FOPD \cite{Chambolle2011fopd}, ADMM~\cite{Boyd2011admm}, to the proposed KS-based variable splitting methods. We also evaluate the performance without adding an extra analysis-$L_1$-regularization term (i.e., $\lambda = 0$) in which case the optimization problem \eqref{eq:x_kf_linear} can be solved by KS. Fig.~\ref{fig:1d_iteration} (a) shows the number of iterations required to solve the estimation problem. In tens of iteration numbers, all the methods have roughly the same relative errors. Fig.~\ref{fig:1d_iteration} (b) shows the average CPU time as function of number of iterations. Note that in order to speed up the KS-based variable splitting methods, we compute the gains $\mathbf{K}_t^y$,  $\mathbf{K}_t^z$ and  $\mathbf{G}_t$ only at the first iteration, and use the pre-computated matrices in the following iterations. Although PRS and \mbox{KS-PRS}, SBM and \mbox{KS-SBM}, FOPD and \mbox{KS-FOPD}, ADMM and \mbox{KS-ADMM} have the same convergence speed, not surprisingly, \mbox{KS-SBM}, \mbox{KS-FOPD} and \mbox{KS-ADMM} have a lower CPU time. When $T = 100 $,  KS and the KS-based variable splitting methods have similar CPU time, but KS has a worse relative error. As shown in Fig.~\ref{fig:1d_iteration} (b), running the PRS, SBM, FOPD and ADMM solvers is time-consuming. The KF-based variable splitting methods take about 0.03 seconds to reach 10 iterations, while the classical optimization approaches such as FOPD and ADMM take $7$ times longer. 
\begin{figure}[ht]
\begin{minipage}[ht]{1\linewidth}
\centering
\includegraphics[width=0.95\textwidth, height=0.65\textwidth]{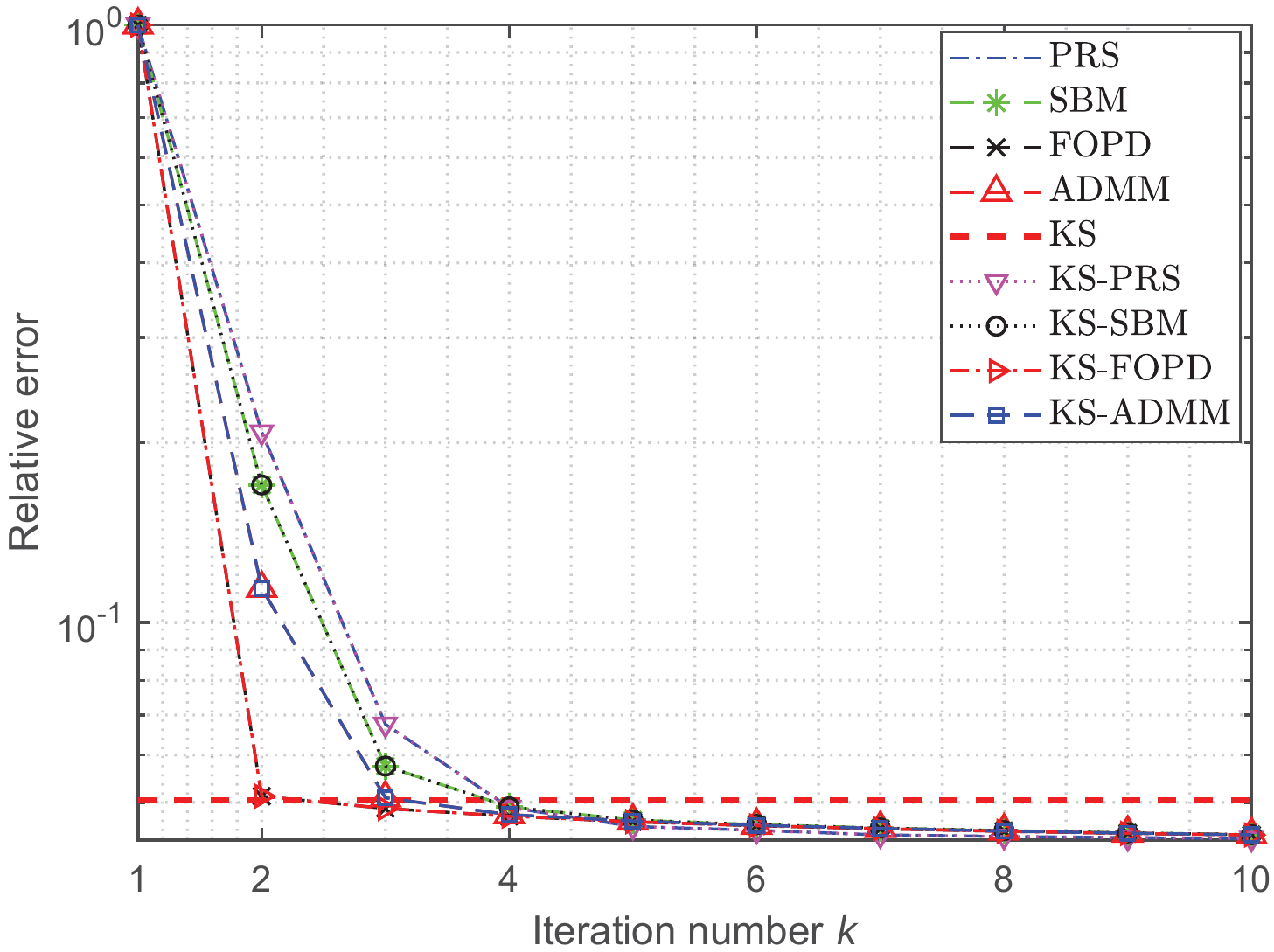}
\label{fig:1d_cov_300}
\centerline{(a)}
\end{minipage}
\\
\begin{minipage}[ht]{1\linewidth}
\centering
\includegraphics[width=0.95\textwidth, height=0.65\textwidth]{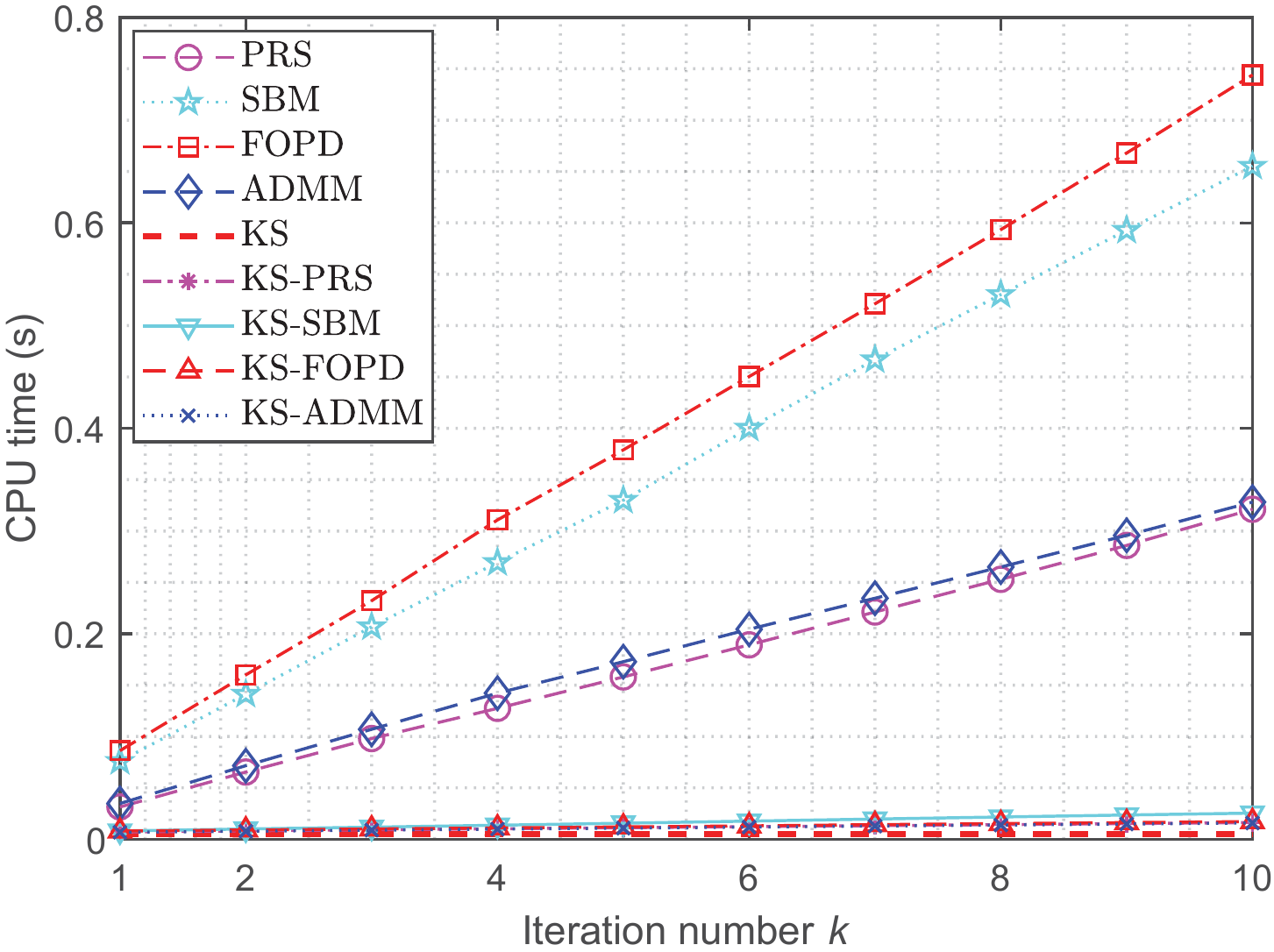}
\label{fig:1d_time}
\centerline{(b)}
\end{minipage}
\caption{Performance as function of iteration number $k$ with $T = 100$ in the linear experiment: (a) relative error versus iteration number, with the $y$-axis in log-scale; (b) average CPU time versus iteration number.}
\label{fig:1d_iteration}
\end{figure}

The benefit of our approach is highlighted by the fact that the methods can efficiently solve a large-scale dynamic signal estimation problem with an extra $L_1$-regularized term. Next, we enlarge the time step count from $10^3$ to $10^8$. All the results reported in Fig.~\ref{fig:1d_varying} are obtained with $\lambda = 1$, which gives the smallest relative error for all the methods. We use $10$ iterations for each method, which in practice is enough for convergence. It can be seen that the proposed method significantly outperforms other batch solutions with respect to the consumed CPU time. In particular, the PRS, SBM, FOPD and ADMM solvers with $10^4$ time steps take more time than the proposed methods with $T = 10^8$. When $T = 10^5, 10^6, 10^7, 10^8$, the computing operations on PRS, SBM, FOPD, and ADMM run out of memory, and the related results cannot be reported. This is mainly because the KS-based variable splitting methods deal with the objective using recursive computations, which significantly reduces the computational and memory burden, while the batch optimization methods explicitly deal with large vectors and matrices.

\begin{figure}[ht]  
\centerline{\includegraphics[width=0.48\textwidth, height=0.32\textwidth]
	{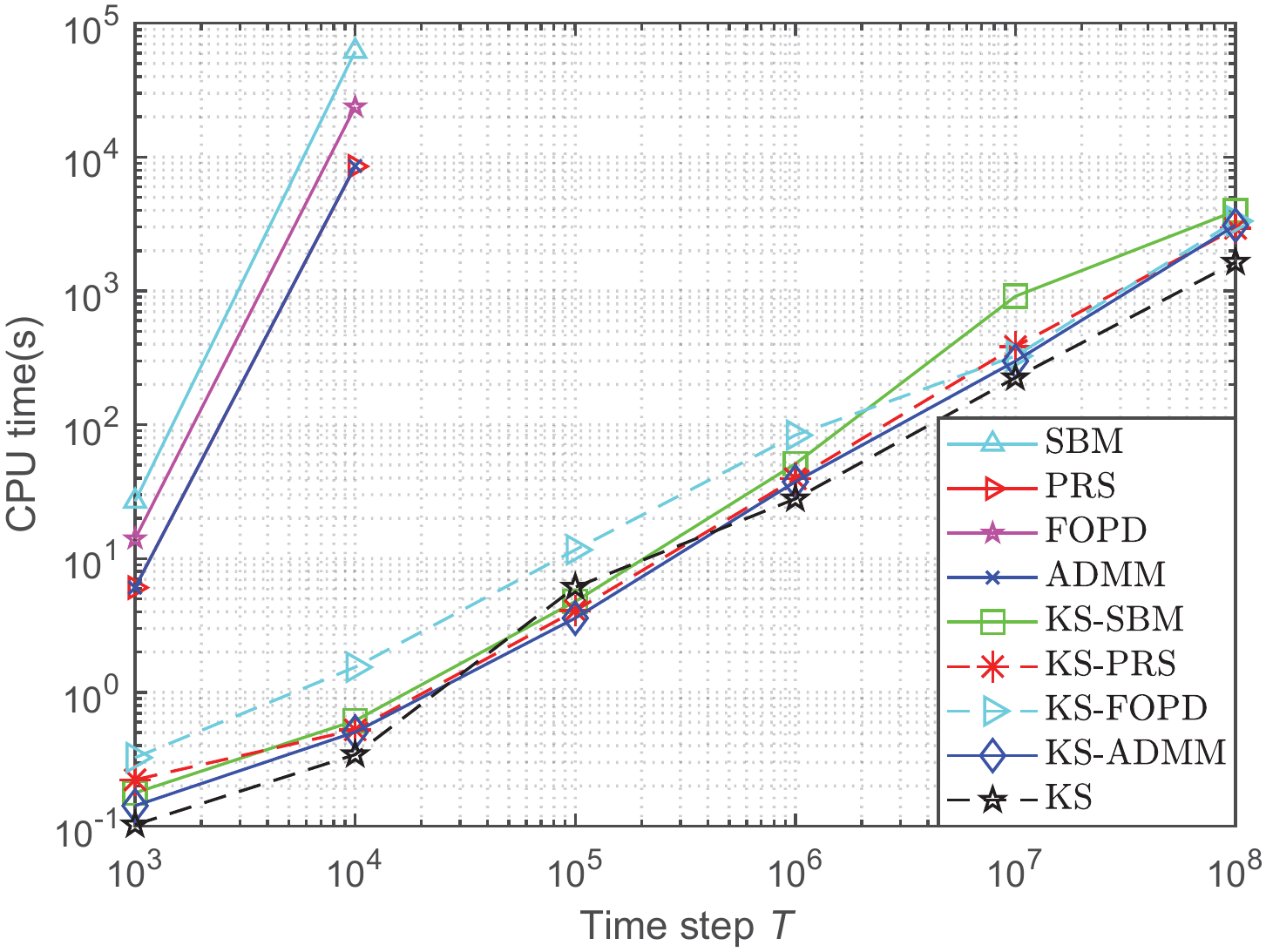}}
\caption{The average CPU time (seconds) versus time steps $T = 10^3$, $10^4$, $10^5$, $10^6$, $10^7$ and $10^8$. The axes are in log-scale.}
\label{fig:1d_varying}
\end{figure}

Table \ref{table:tab2} summarizes the average CPU time with different regularization parameters $\lambda$ when the number of time steps $T$ is varying from $10^2$ to $10^8$. The table reports the average CPU time in seconds required by each solver with $10$ iterations. Not surprisingly, the KF-based variable splitting methods (i.e., KF-PRS, KF-SBM, KF-FOPD, and KF-ADMM) are much faster than the batch variable splitting methods when $T$ grows.

\begin{table*}[ht]
\caption{Comparison of average CPU time (seconds) with different $\lambda$ in the linear case.}
	\begin{center}
		 \centering
	\begin{tabular}{|c|c|c|c|c|c|c|c|c|c|}
				\hline
	    {${\lambda}$} & {${T}$} & PRS & SBM & FOPD  & ADMM   & KS-PRS & KS-SBM & KS-FOPD & KS-ADMM \\
	    \hline
\multirow{6}*{0.1} &{${10^3}$} &  $\bm{6.05}$ & 25.42 & 14.06 & 6.12 & 0.23 & 0.16 &  0.31 &  0.16 \\
\cline{2-10}
~&{${10^4}$} & 841 & 6210 & 2379  & 851   & 0.54 & 0.60 &  1.55  &  0.53   \\
\cline{2-10}
~&	{${10^5}$} & - &-  & -  & -   & 28.1 & 30.1  & 11.2 & 9.7    \\
\cline{2-10}
~&{${10^6}$}& - &-  & -  & -  & 39.4  & 52.1  & 87.8  & 38.5  \\
\cline{2-10}
~&	{${10^7}$} &- &-  & -  & -   & 402 & 922 & 341 & 337   \\
\cline{2-10}
~&{${10^8}$} &- &-  & -  & -   & 3330 & 4121 & 3510 &  3378  \\
			  \hline
            \multirow{6}*{0.5} &{${10^3}$} &  6.07 & 24.61 & 14.04 & 6.06 & 0.22 & 0.17 &  0.32 &  0.14 \\
\cline{2-10}
~&{${10^4}$} & 832 & 6178 & 2366 & 837   &0.52 & 0.60 &  $\bm{1.53}$  &  0.50   \\
\cline{2-10}
~&	{${10^5}$} & - &-  & -  & -   & 27.9 & 30.0  & 10.9 & 9.4    \\
\cline{2-10}
~&{${10^6}$}& - &-  & -  & -  & 38.9  &51.3  & 87.2  & 38.1  \\
\cline{2-10}
~&	{${10^7}$} &- &-  & -  & -   & 391 & 912 & 337 & 312   \\
\cline{2-10}
~&{${10^8}$} &- &-  & -  & -  & 3121 & 4003 & 3421 &  3115  \\
			  \hline
\multirow{6}*{1} &{${10^3}$} &  6.06 & $\bm{23.65}$ &  $\bm{14.04}$ & $\bm{6.04}$ & ${0.22}$ & $\bm{0.15}$ &  $\bm{0.31}$ &   $\bm{0.14}$ \\
\cline{2-10}
~&{${10^4}$} & $\bm{814}$ & $\bm{5943}$ & $\bm{2189}$ & $\bm{824}$   & $\bm{0.52}$ & $\bm{0.59}$ &  1.56  & $\bm{ 0.47}$   \\
\cline{2-10}
~&	{${10^5}$} & - &-  & -  & - & $\bm{27.4}$ & $\bm{29.8}$  & $\bm{10.7}$ & $\bm{9.2}$     \\
\cline{2-10}
~&{${10^6}$}& - &-  & -  & -   & $\bm{37.4}$  & $\bm{ 48.8}$  & $\bm{83.1}$  & $\bm{37.3}$   \\
\cline{2-10}
~&{${10^7}$} &- &-  & -  & -   &  $\bm{383}$ & $\bm{889}$ & $\bm{326}$  & $\bm{298}$   \\
\cline{2-10}
~&{${10^8}$} &- &-  & -  & -   & $\bm{2936}$ & $\bm{3961}$ & $\bm{3328}$ &  $\bm{3096}$  \\
          \hline
        \multirow{6}*{2} &{${10^3}$} &  6.12 & 28.64 & 16.02 & 6.16 & $\bm{0.21}$ & 0.18 &  0.32 &  0.14 \\
\cline{2-10}
~&{${10^4}$} & 874 & 6219 & 2415 & 868   &0.59 &0.62 &  1.54  &  0.51   \\
\cline{2-10}
~&	{${10^5}$} & - &-  & -  & -   & 30.0 & 31.9  & 12.1 & 9.2    \\
\cline{2-10}
~&{${10^6}$}& - &-  & -  & -  & 40  & 53  & 86  & 38  \\
\cline{2-10}
~&	{${10^7}$} &- &-  & -  & -  & 417 & 1002 & 345 & 316   \\
\cline{2-10}
~&{${10^8}$} &- &-  & -  & -   & 3148 & 4021 & 3510 &  3278  \\
     \hline
	\end{tabular}
\end{center} \label{table:tab2}
\end{table*}

\subsection{Nonlinear Simulation Experiment}
\label{results:nonlinear}
We consider a five-dimensional nonlinear coordinated turn model~\cite{Bar-Shalom+Li+Kirubarajan:2001}. We set the measurement noise covariance $\mathbf{R}_t = \sigma^2 \mathbf{I}$ with $\sigma = 0.3$,  $q_c = 0.01$, $\Delta t = 0.2$, and run $k_{\max}=20$ iterations of all the optimization methods. Before moving on, we provide some empirical evidence to support the choice of the regularization parameter $\lambda$ in the IEKS-based variable splitting methods. Similarly to the linear case in Section \ref{results:linear}, we plot the relative errors obtained by varying $\lambda$ in Fig.~\ref{fig:ct_err_lam}. It can be seen that the IEKS-based variable splitting methods have similar relative errors with $\lambda$ varying in the range $[0.1, \, 1]$. Next, we select $\lambda = 0.1$ for the following experiments.
\begin{figure}[!htb]  
	\centerline{\includegraphics[width=0.46\textwidth, height=0.32\textwidth]
		{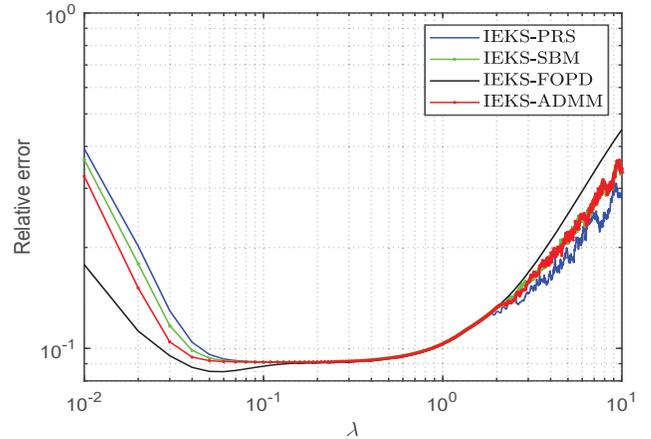}}
	\caption{Relative error as function of regularization parameter $\lambda$ in the five-dimensional nonlinear coordinated turn model.}
	\label{fig:ct_err_lam}
\end{figure}

Then, we compare IEKS\cite{Bell1994smoother}, \mbox{GN-PRS} and \mbox{IEKS-PRS}, \mbox{GN-SBM} and \mbox{IEKS-SBM}, \mbox{GN-FOPD}, and \mbox{IEKS-FOPD},  \mbox{GN-ADMM}, and \mbox{IEKS-ADMM} by plotting the relative error and the CPU time as functions of the iteration number. 
Fig.~\ref{fig:ct_1d_iteration} demonstrates the efficiency of our IEKS-based variable splitting methods against \mbox{GN-PRS}, \mbox{GN-SBM}, \mbox{GN-FOPD} and \mbox{GN-ADMM} in the same experiment. The horizontal axis in Fig.~\ref{fig:ct_1d_iteration} (a) describes the total iteration number, and the vertical axis gives the relative errors.  As can be seen, all the methods give the fast convergence in around $5$ iterations.

We are also interested in the performance without adding an extra analysis-$L_1$-regularized term (i.e., $\lambda = 0$). Since there is no outer iteration in IEKS, we plot the relative error of IEKS after the inner iteration  (dashed black line in Fig.~\ref{fig:ct_1d_iteration} (b)). In contrast with the estimation results, we observe a performance gap between the variable splitting methods and \mbox{IEKS}. This gap reveals the benefit of the extra regularization term that is used in these methods.
In the average CPU time,  \mbox{IEKS-PRS}, \mbox{IEKS-SBM}, \mbox{IEKS-FOPD}, and \mbox{IEKS-ADMM} are clearly superior to batch variable splitting (see Fig.~\ref{fig:ct_1d_iteration} (b)). \mbox{IEKS-FOPD} is the fastest convergent method, needing only $3$ iterations.  Although the state estimation problem is relatively small scale, \mbox{GN-PRS}, \mbox{GN-SBM}, \mbox{GN-FOPD} and \mbox{GN-ADMM} are still time-consuming.
\begin{figure}
\begin{minipage}[t]{1\linewidth}
\centering
\includegraphics[width=0.95\textwidth, height=0.65\textwidth]{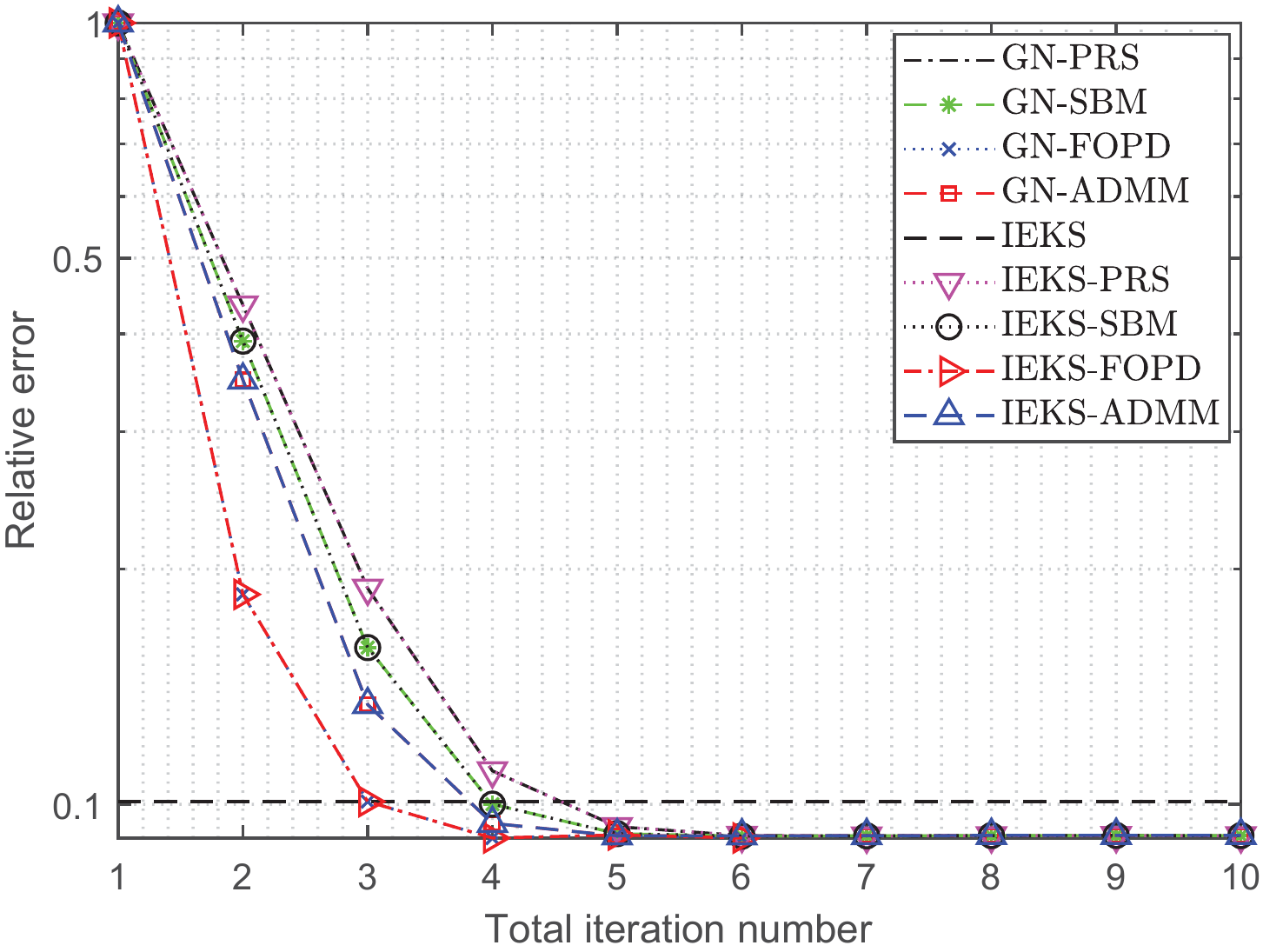}
  \centerline{(a)}
\end{minipage}
\\
\begin{minipage}[t]{1\linewidth}
\centering
\includegraphics[width=0.95\textwidth, height=0.65\textwidth]{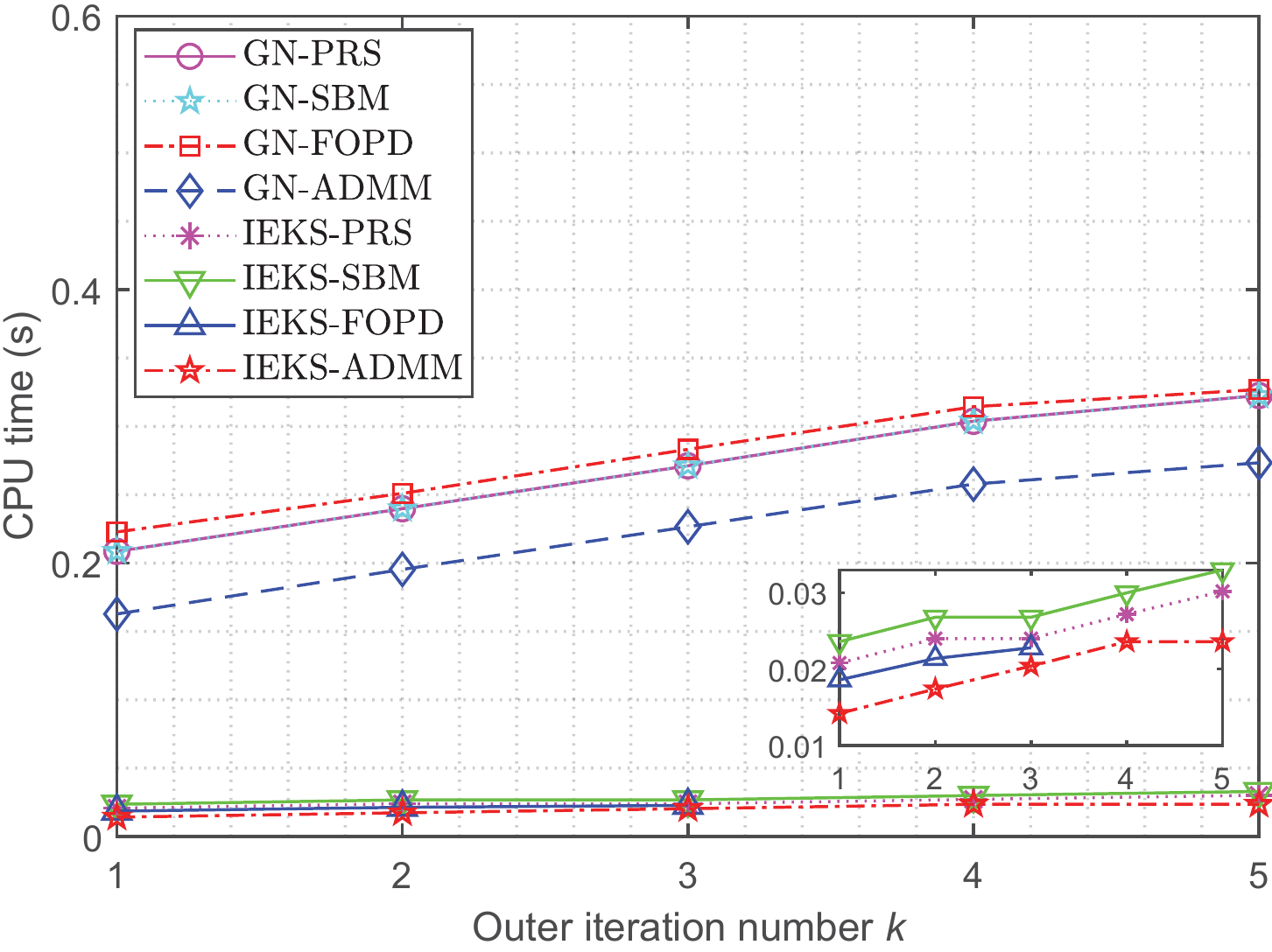}
  \centerline{(b)}
\end{minipage}
\caption{Performance in the coordinated turn model: (a) relative error versus total iteration number, with the $y$-axis in log-scale; (b) average CPU time (seconds) versus outer iteration number $k$.}
\label{fig:ct_1d_iteration}
\end{figure}

Like with linear tracking model, also in this simulation, we enlarge the time step count $T$ from $10^2$ to $10^7$, and plot the results in Fig.~\ref{fig:ct_varying_cpu}. When increasing the time step count, our proposed methods significantly outperform the batch methods with respect to CPU time. In particular, the PRS, SBM, FOPD and ADMM solvers with $10^4$ time steps take more time than our proposed methods with $10^7$ time steps, when $\lambda \in[0.1, 0.5, 1, 2] $. In Table \ref{table:tab3}, we further list the CPU times with different $\lambda$ when $T$ is varying.  The performance benefit of the proposed methods is evident.
\begin{figure}[h!]  
\centerline{\includegraphics[width=0.48\textwidth, height=0.32\textwidth] {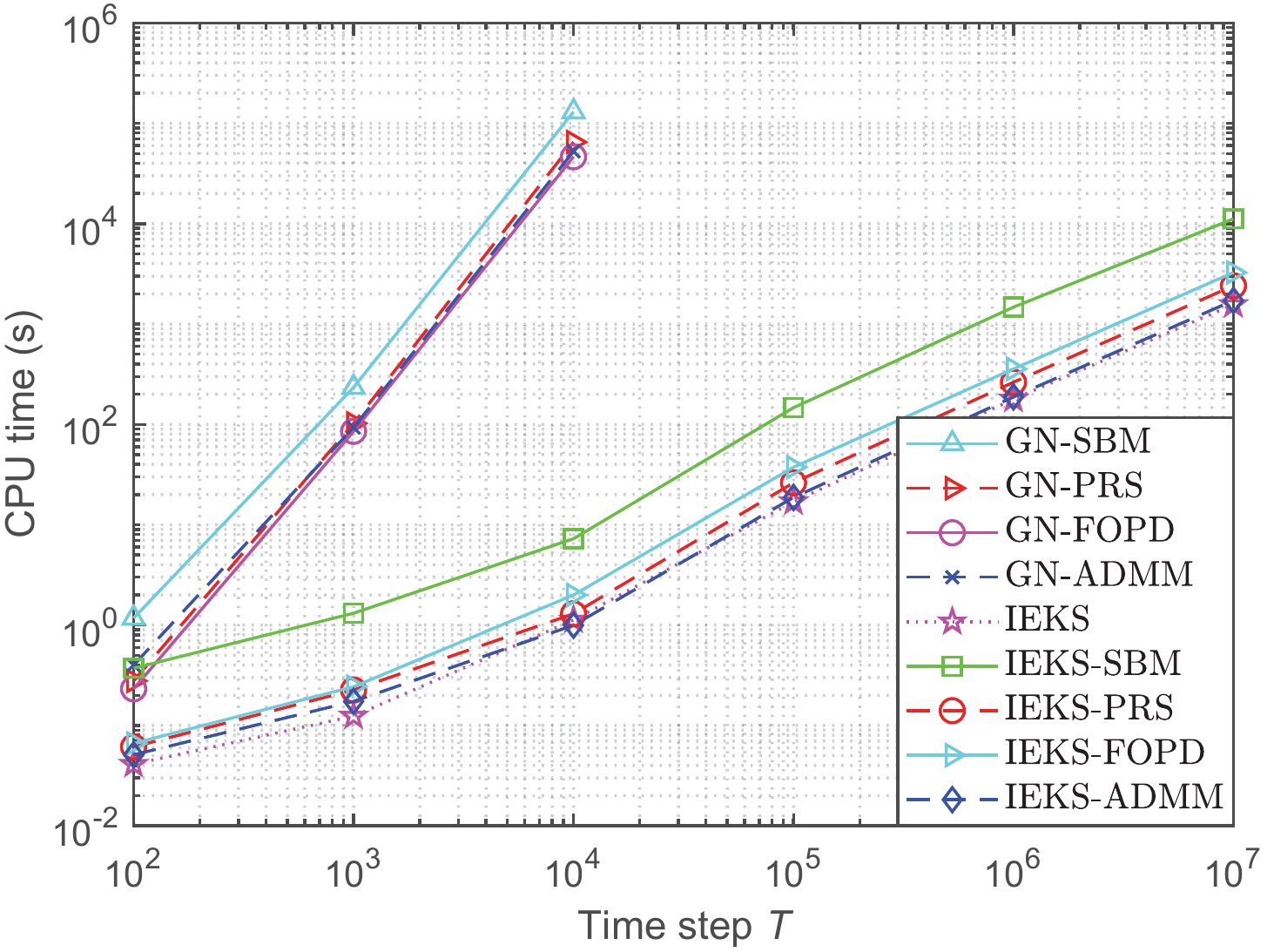}}
\caption{Average CPU time (seconds) versus time step $T$ is $10^2$, $10^3$, $10^4$, $10^5$, $10^6$, $10^7$ in the nonlinear simulated trajectory. The axes are in log-scale.}
\label{fig:ct_varying_cpu}
\end{figure}

\begin{table*}[h!]
	\caption{Comparison of average CPU time (seconds) with different $\lambda$ in the nonlinear case.}
	\begin{center}
		\centering
		\begin{tabular}{|c|c|c|c|c|c|c|c|c|c|}
			\hline
			{${\lambda}$} & {${T}$} & GN-PRS & GN-SBM&GN-FOPD&GN-ADMM &IEKS-PRS&IEKS-SBM &IEKS-FOPD  &IEKS-ADMM \\
			\hline
			\multirow{6}*{0.01} &{${10^2}$} & $\bm{0.26}$ & $\bm{1.15}$  & 0.24  &  $\bm{0.38}$  & 0.07 & 0.39  &0.08 &0.05 \\
			\cline{2-10}
			~&{${10^3}$} & $\bm{99}$  & $\bm{231}$  & 86  &  94  & 0.22  & 1.32 & 0.25 & 0.19 \\
			\cline{2-10}
			~&	{${10^4}$} & 1504  & 5624 &  4951  & 1396  & 1.31  & 7.27 & 1.16  & 2.01\\
			\cline{2-10}
			~&{${10^5}$}& - & - & - & - & 26.43 &  146.94 & 38.12 & 19.01\\
			\cline{2-10}
			~&	{${10^6}$} & - & - & - & - & 275 & 1687 & 389 & 201\\
			\cline{2-10}
			~&{${10^7}$} & - & - & - & - & 2621 & 11542 & 3460 & 1963 \\
			\hline
			\multirow{6}*{0.1} &{${10^2}$} &0.27 & 1.17 & $\bm{0.23}$  &  0.39  & $\bm{0.06}$ & $\bm{0.37}$  & $\bm{0.07}$ & $\bm{0.05}$  \\
			\cline{2-10}
			~&{${10^3}$} & 102  & 233  & $\bm{85}$  & $\bm{93}$   & $\bm{ 0.22}$  & $\bm{1.31}$  & $\bm{0.24}$  & $\bm{0.17}$  \\
			\cline{2-10}
			~&	{${10^4}$} & $\bm{1303}$  & $\bm{5212.7}$ & $\bm{4640}$  & $\bm{1252.6}$ & $\bm{1.29}$  & $\bm{7.26}$ & $\bm{1.16}$  &  $\bm{1.99}$    \\
			\cline{2-10}
			~&{${10^5}$}& - &-  & -  & -   & $\bm{26.21}$ & $\bm{146.81}$ &  $\bm{37.45}$ & $\bm{18.72}$ \\
			\cline{2-10}
			~&	{${10^6}$} &- &-  & -  & -    & $\bm{263}$ & $\bm{1473}$ & $\bm{356}$ &  $\bm{187}$ \\
			\cline{2-10}
			~&{${10^7}$} &- &-  & -  & -   &   $\bm{2402}$  &  $\bm{11155}$ & $\bm{3260}$ &  $\bm{1716}$   \\
			\hline
			\multirow{6}*{1} &{${10^2}$} & 0.32 & 1.28 & 0.26 & 0.43 & 0.06 & 0.39 & 0.07 & 0.05 \\
			\cline{2-10}
			~&{${10^3}$} & 105 & 241 & 89 & 102 & 0.23  & 1.38 & 0.31 & 0.26 \\
			\cline{2-10}
			~&	{${10^4}$} & 1597 & 5711 & 5001 & 1450 & 1.32 & 7.31 & 1.19  & 2.03 \\
			\cline{2-10}
			~&{${10^5}$}& - & - & - & - & 27.21 & 152.1  & 39.45 & 20.12 \\
			\cline{2-10}
			~&	{${10^6}$} & - & - & - & - & 291  & 1492 & 361 & 202  \\
			\cline{2-10}
			~&{${10^7}$} & - & - & - & - & 2645 & 11784 & 3519 & 2001 \\
			\hline
		   \multirow{6}*{2} &{${10^2}$} & 0.34 & 1.28 & 0.27 & 0.43 & 0.06 & 0.43 & 0.07 & 0.07 \\
			\cline{2-10}
			~&{${10^3}$} & 111 & 249 & 95 & 121 & 0.26  & 1.41 & 0.34 & 0.27 \\
			\cline{2-10}
			~&	{${10^4}$} & 1612 & 5821 & 5294 & 1510 & 1.48 & 7.48 & 1.23  & 2.45 \\
			\cline{2-10}
			~&{${10^5}$}& - & - & - & - & 28.45 & 167 & 41.24 & 22.12   \\
			\cline{2-10}
			~&	{${10^6}$} & - & - & - & - & 312 &1625 & 389 & 241 \\
			\cline{2-10}
			~&{${10^7}$} & - & - & - & - & 2741 & 12016 & 3645 & 2268 \\
			\hline
		\end{tabular}
	\end{center} \label{table:tab3}
\end{table*}

\subsection{Tomographic Reconstruction}
\label{sec:tomographic}
In this section, we consider the application of the methodology to X-ray computed tomography (CT) imaging \cite{Pfister2014ct,Bubba2017ct}. First, we evaluate the performance of the proposed methods on real tomographic X-ray data of an emoji phantom measured at the University of Helsinki \cite{Meaney2018emoji}. The dataset consists of 33-point time series of the X-ray sinogram of an emoji made of small squared ceramic stones. In the sequence, the emoji transforms from a face with closed eyes and a straight mouth to a face with smiling eyes and mouth. Typically, we have a sequence of square X-ray images of size $s \times s$ with $s = 64,128$, which we are interested in reconstructing from low-dose observations taken from a limited number of angles. These low-dose observations can be modeled by the measurement matrix $\mathbf{H}_t$ which describes line integrals through the object (i.e., Radon transform). 

\begin{figure}[h!]
\begin{minipage}{0.32\linewidth}
  \centerline{\includegraphics[width=0.96\textwidth, height=0.7\textwidth]{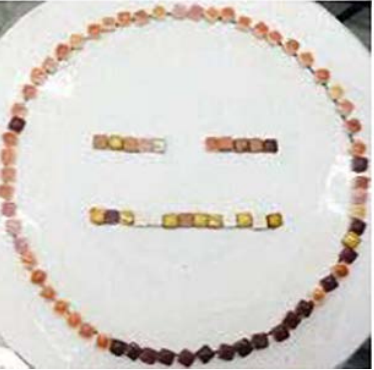}}
\end{minipage}
\begin{minipage}{0.32\linewidth}
  \centerline{\includegraphics[width=0.96\textwidth, height=0.7\textwidth]{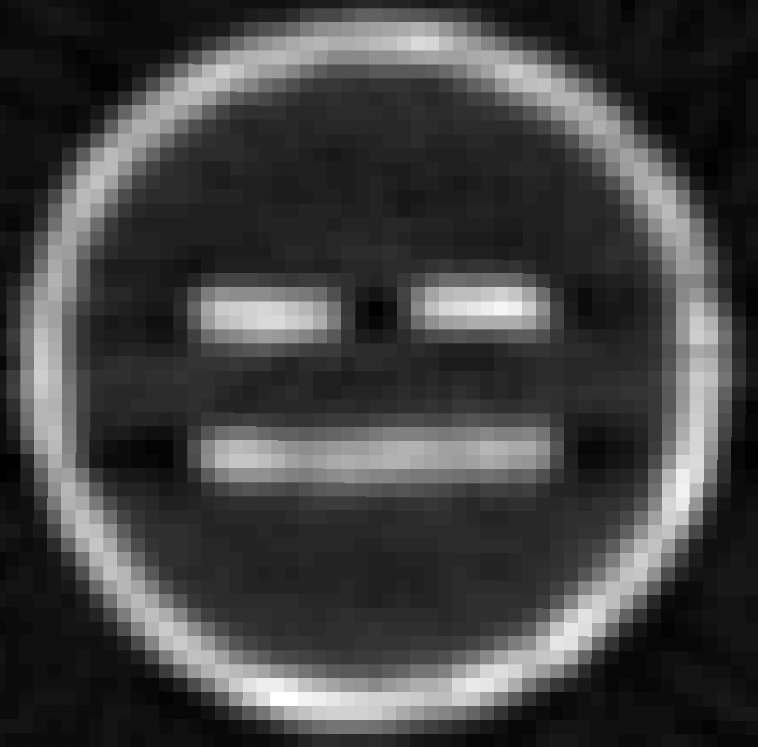}}
\end{minipage}
\begin{minipage}{0.32\linewidth}
  \centerline{\includegraphics[width=0.96\textwidth, height=0.7\textwidth]{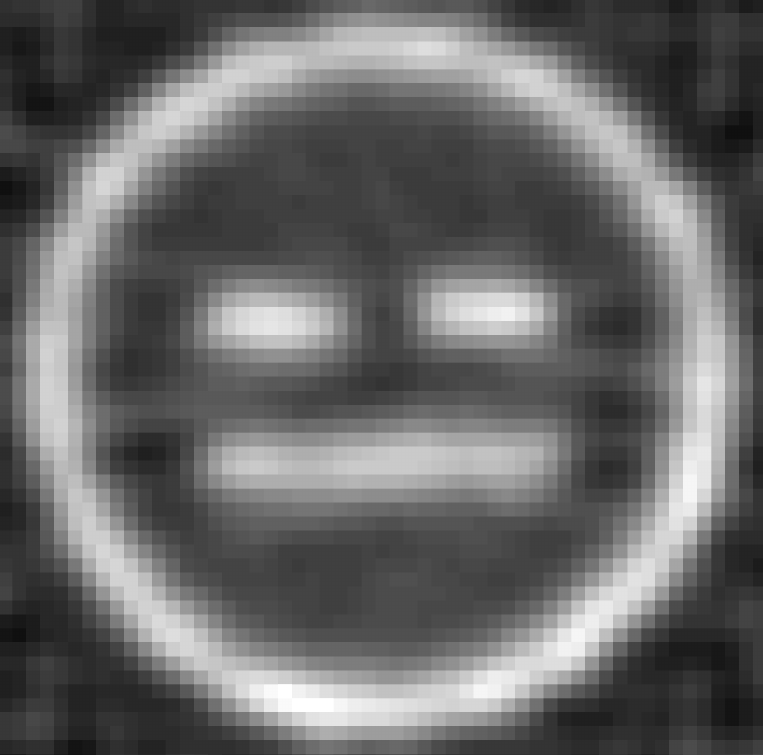}}
\end{minipage}
\centering 
\vfill
\vspace{5pt}
\begin{minipage}{0.32\linewidth}
  \centerline{\includegraphics[width=0.96\textwidth, height=0.7\textwidth]{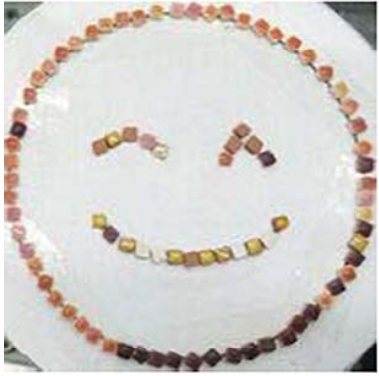}}
\end{minipage}
\begin{minipage}{0.32\linewidth}
  \centerline{\includegraphics[width=0.96\textwidth, height=0.7\textwidth]{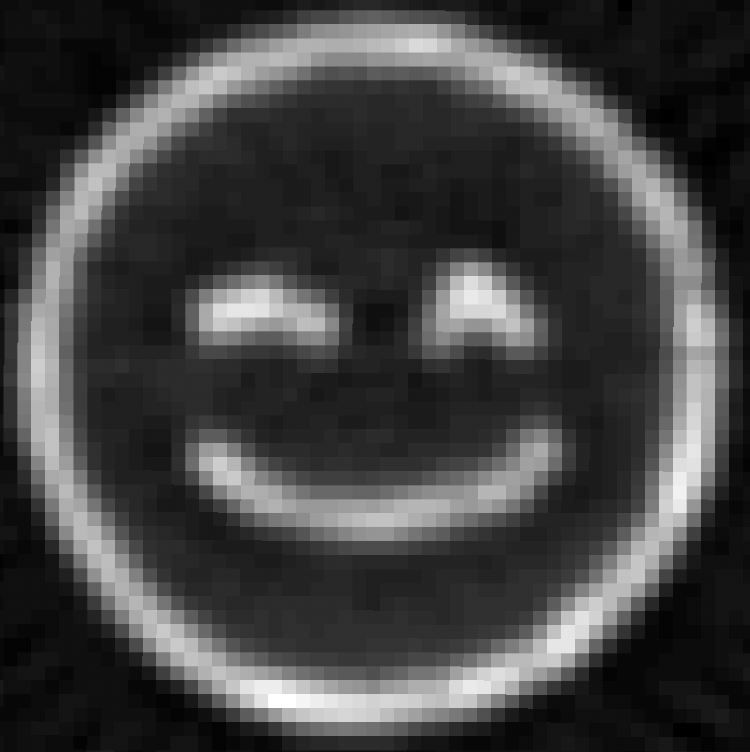}}
\end{minipage}
\begin{minipage}{0.32\linewidth}
  \centerline{\includegraphics[width=0.96\textwidth, height=0.7\textwidth]{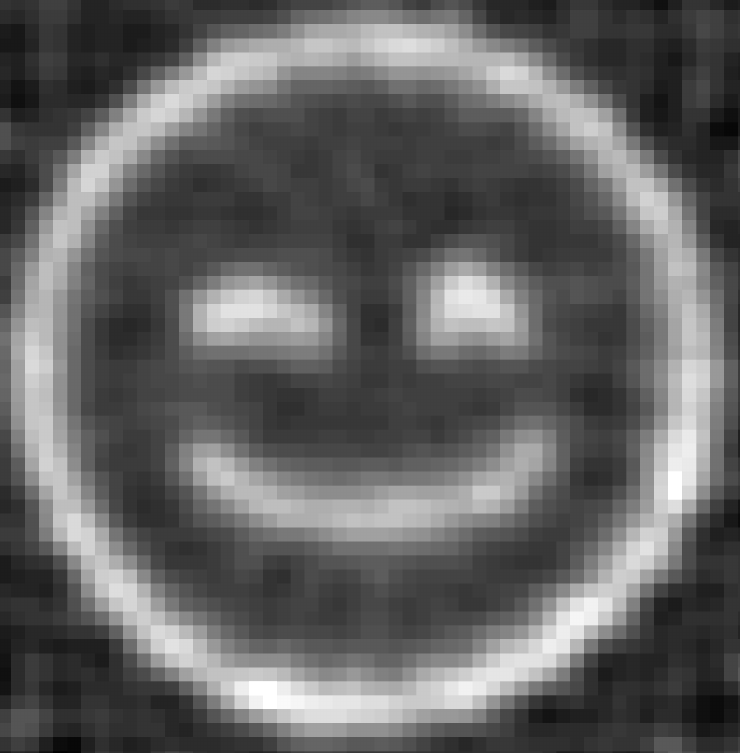}}
\end{minipage}
\centering
\vfill
\vspace{5pt}
\begin{minipage}{0.32\linewidth}
  \centerline{\includegraphics[width=0.96\textwidth, height=0.7\textwidth]{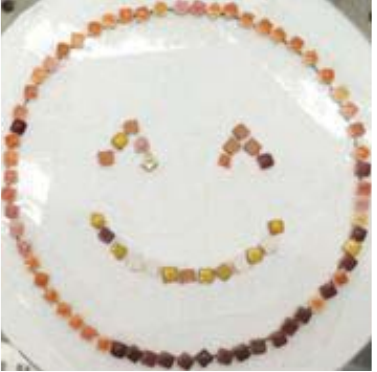}}
\end{minipage}
\begin{minipage}{0.32\linewidth}
  \centerline{\includegraphics[width=0.96\textwidth, height=0.7\textwidth]{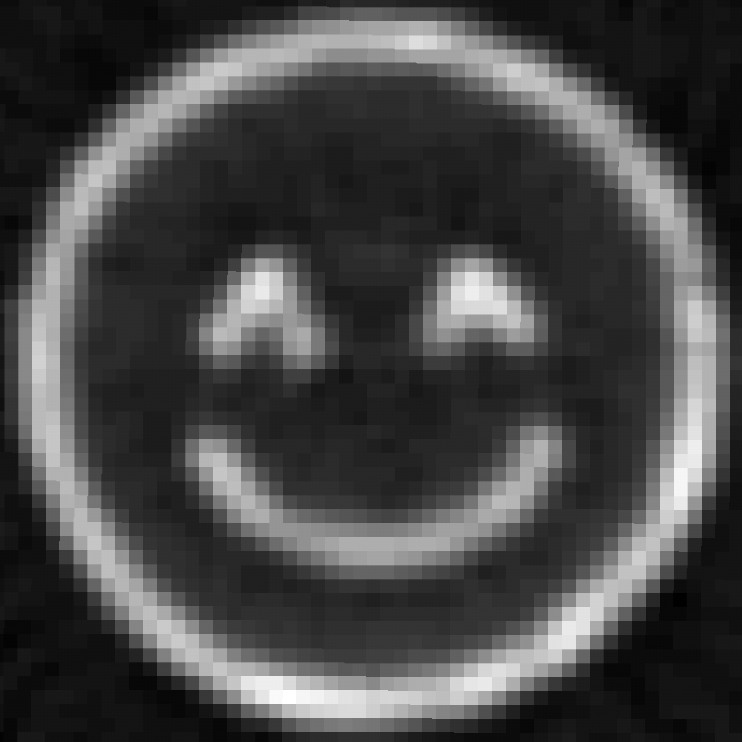}}
\end{minipage}
\begin{minipage}{0.32\linewidth}
  \centerline{\includegraphics[width=0.96\textwidth, height=0.7\textwidth]{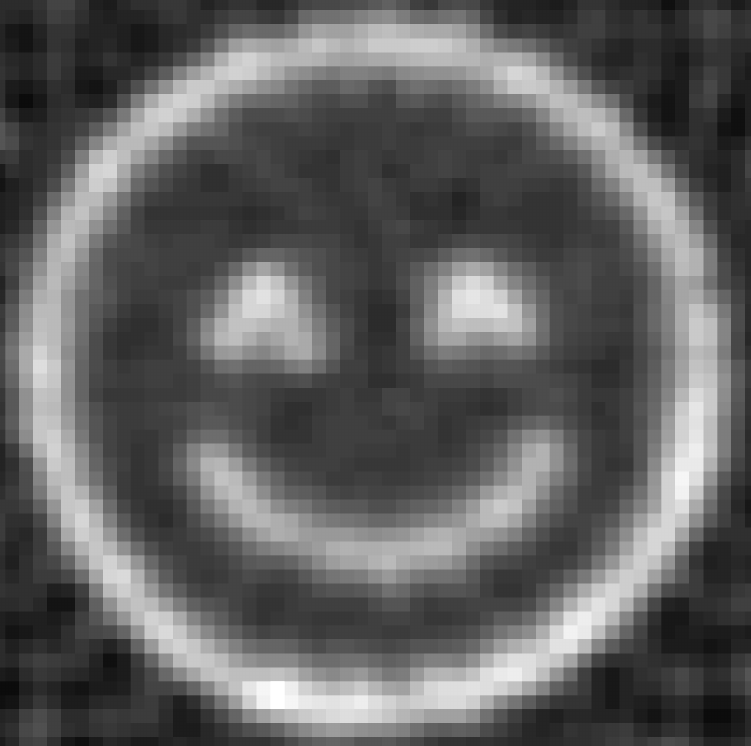}}
\end{minipage}
\centering
\vfill
\caption{Reconstruction results for the emoji motion dataset. Original pictures at time step $t=1,15,33$ are given in the first column, respectively. Reconstruction results from $60$ and $30$ projections are shown in the middle column and the third column.}
\label{fig:ct_Emoji}
\end{figure}

Fig.~\ref{fig:ct_Emoji} shows the CT reconstruction results for the emoji motion dataset, obtained by KS-ADMM. We set the parameters to $\lambda = 10$, $\rho = 1$, $k_{\max} = 20$, and $n_x = 4096$. The analysis operator $\mathbf{\Omega}_t$ consists of all the vertical and horizontal gradients (one step differences), which corresponds to so called TV regularization \cite{Hu2012TV}. The number of measurements that correspond to $60$ or $30$ projections are $n_y = 13020$ and $n_y =6510$, respectively. Although there is no ground truth to compare the qualitative results, we can observe the visual results from different numbers of projections. When the number of projections is $60$, the method provides good reconstruction results with $20$ iterations. We see that the 30-projection results suffer from the block artifacts as a consequence of the reduction in dose. 

Furthermore, we validate the effectiveness of the proposed method on the real inhalation (iBH-CT) and exhalation (eBH-CT) breath-hold CT images, which was acquired as part of the National Heart Lung Blood Institute COPDgene study \cite{Castillo2009CT}. The dataset consists of 10 expiratory phase images of the segmented lung voxels. In detail, the parameters are $\lambda = 1$, $\rho = 0.1$, $k_{\max} = 15$, $T = 10$, $n_x = 16384$, and the numbers of measurements are $n_y = 13020$ and $n_y =6510$, corresponding to $60$ and $20$ projections. The ground truth and the reconstruction results are shown in Fig.~\ref{fig:ct_medical}. By visually comparing the results, we observe that moving from $60$ to $20$ projection provides much more drastic change. For example, some additional artifacts exist, but the result with the setting $n_x = 16384$ and $n_y =6510$ is still very much acceptable (see the third column in Fig.~\ref{fig:ct_medical}). The results show that our methods still successfully preserve temporal information when the number of projections is $20$.

\begin{figure}[!htb]
	\begin{minipage}{0.32\linewidth}
		\centerline{\includegraphics[width=0.96\textwidth, height=0.7\textwidth]{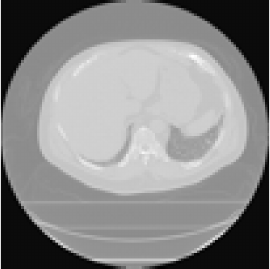}}
	\end{minipage}
	\begin{minipage}{0.32\linewidth}
		\centerline{\includegraphics[width=0.96\textwidth, height=0.7\textwidth]{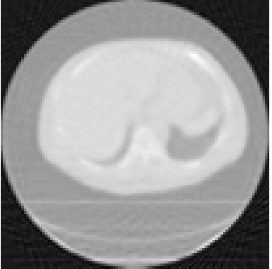}}
	\end{minipage}
	\begin{minipage}{0.32\linewidth}
		\centerline{\includegraphics[width=0.96\textwidth, height=0.7\textwidth]{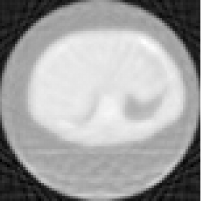}}
	\end{minipage}
	\centering 
	\vfill
	\vspace{5pt}
	\begin{minipage}{0.32\linewidth}
		\centerline{\includegraphics[width=0.96\textwidth, height=0.7\textwidth]{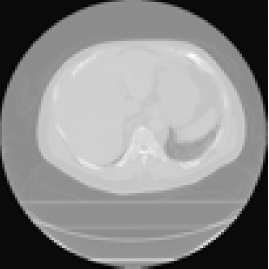}}
	\end{minipage}
	\begin{minipage}{0.32\linewidth}
		\centerline{\includegraphics[width=0.96\textwidth, height=0.7\textwidth]{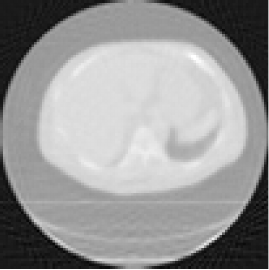}}
	\end{minipage}
	\begin{minipage}{0.32\linewidth}
		\centerline{\includegraphics[width=0.96\textwidth, height=0.7\textwidth]{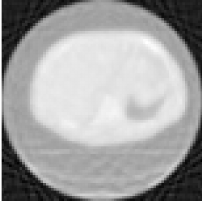}}
	\end{minipage}
	\centering
	\vfill
	\vspace{5pt}
	\begin{minipage}{0.32\linewidth}
		\centerline{\includegraphics[width=0.96\textwidth, height=0.7\textwidth]{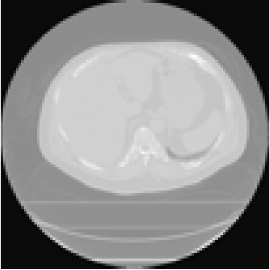}}
	\end{minipage}
	\begin{minipage}{0.32\linewidth}
		\centerline{\includegraphics[width=0.96\textwidth, height=0.7\textwidth]{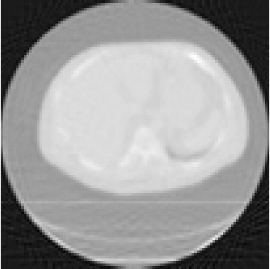}}
	\end{minipage}
	\begin{minipage}{0.32\linewidth}
		\centerline{\includegraphics[width=0.96\textwidth, height=0.7\textwidth]{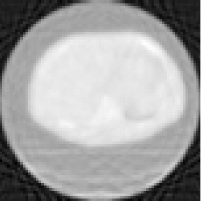}}
	\end{minipage}
	\vfill
	\caption{Reconstruction results for the lung dataset. Original pictures at time step $t=1,3,6$ are given in the first column. Reconstruction results from $60$ and $20$ projections are shown in the middle column and the third column.}
	\label{fig:ct_medical}
\end{figure}

In the two experiments, we used a stationary Kalman filter and smoother to implement the optimization. We pre-computed all the gains before the iteration, which significantly speeded up the computations in tomographic reconstruction. We report CPU time (seconds) in Table~\ref{tab:ct_time}. Table~\ref{tab:ct_time} shows that KS-ADMM achieves significantly lower CPU time than the batch ADMM although the visual quality of all the reconstructions is equal. For example, in emoji motion dataset, when $n_x = 16384$ , ADMM takes three time longer than our proposed method. In the lung dataset, when $n_x = 16384$ and $n_y = 6510$, \mbox{KS-ADMM} seems to be promising to provide computationally efficient reconstruction.

\begin{table}[htb!]
\caption{Average CPU time (seconds) for the outer iteration in the tomographic recostruction.}
\begin{center}
 \centering
 \small
\begin{tabular}{|c|c|c|c|c|c|}
\hline
Dataset  & $T$ & $n_x$  & $n_y$ & ADMM &  KF-ADMM  \\ 
\hline
\multirow{4}*{Emoji} & \multirow{4}*{$33$} &16384 & 13020 & 3657.58 & 771.49 \\ 
\cline{3-6}
~  &~  &16384 & 6510  & 1422.38 & 387.26 \\ 
\cline{3-6}
~   &~  &4096 & 13020 & 284.83 & 97.93 \\ 
\cline{3-6}
~   &~  &4096 & 6510 & 77.79 & 20.10 \\ 
\hline
\hline
\multirow{4}*{Lung} & \multirow{4}*{$10$} &16384 & 13020 & 3584.62 & 764.78 \\ 
\cline{3-6}
~  &~  &16384 & 6510  & 1378.41 & 367.13 \\ 
\cline{3-6}
~   &~  &10404 & 13020 & 1475.84 & 487.25 \\ 
\cline{3-6}
~   &~  &10404 & 6510 & 187.87 & 61.24 \\ 
\hline
\end{tabular}
\label{tab:ct_time}
\end{center}
\end{table}

\label{results:tomographic}
\section{Conclusion}
\label{sec:conclusion}

In this paper, we have presented two new classes of methods for solving state estimation problems. The estimation problem has been formulated as an (analysis) $L_1$-regularized optimization problem and the resulting problem has been solved by using the combinations of (iterated extended) Kalman smoother and variable splitting methods such as ADMM. The proposed approaches replace the batch solution for the state-update by using the smoother, which has a lower time-complexity than the batch solution. Furthermore, we have extended the proposed methods to a more general algorithmic framework, where the state-update is computed with the smoother. We have also established (local) convergence results for the novel KS-ADMM and IEKS-ADMM methods. In two different linear and nonlinear simulated cases, we have presented experimental results which show the efficiency of the smoother-based variable splitting optimization methods, especially when applied to large-scale or high-dimensional $L_1$-regularized state estimation problems. We also applied the methodology to a real-life tomographic reconstruction problem arising in X-ray-based computed tomography. 
Further work may explore a proper choice of the dual parameters using in KS / IEKS-based variable splitting methods, and discuss the convergence in the adaptive parameter settings.

\section*{Acknowledgements}
The authors are grateful for the help of Zenith Purisha in preparing the computed tomography experiment and Zheng Zhao for useful comments on the manuscript.

\appendices
\section{Proof of Lemma 1}
\label{pf:lemma:nonincreasing}
We first prove for \textbf{Case (a)}. By the first-order optimality condition of $\mathbf{x}$ subproblem, we have 
	\begin{equation}
	\label{eq:first_order_x}
	\begin{split}
	\begin{aligned}
	&\nabla \mathcal{L}(\mathbf{x}^{(k+1)},\mathbf{w}^{(k)};\bm{\eta}^{(k)}) = 0, 
	\end{aligned}
	\end{split} 
	\end{equation}
which implies that 
	\begin{equation}
	\label{eq:first_order_x}
	\begin{split}
	\begin{aligned}
\nabla \theta_1(\mathbf{x}^{(k+1)})  &=   \mathbf{\Omega}^\T  ( \bm{\eta}^{(k)}+ \rho (\mathbf{w}^{(k)} - \mathbf{\Omega}\mathbf{x}^{(k+1)} ))  \\
              &=   \mathbf{\Omega}^\T \bm{\eta}^{(k+1)}.
	\end{aligned}
	\end{split} 
	\end{equation}
It follows that 
\begin{equation}
\label{eq:bound_dual}
\begin{split}
\begin{aligned}
	  \|  \mathbf{\Omega}^\T \bm{\eta}^{(k+1)}  -  \mathbf{\Omega}^\T \bm{\eta}^{(k)} \|  
	  {\leq}   {L_{\theta_1} }  \| \mathbf{x}^{(k+1)} - \mathbf{x}^{(k)}\|. 
\end{aligned}
\end{split} 
\end{equation}
Then, if we assume that $\mathbf{\Omega}$ is full-row rank with $\mathbf{\Omega} \mathbf{\Omega}^\T \succeq \kappa_a^2 \mathbf{I}$, we have 
\begin{equation}
\label{eq:bound_dual_2}
\begin{split}
\begin{aligned}
	 \|  \mathbf{\Omega}^\T \bm{\eta}^{(k+1)}  -  \mathbf{\Omega}^\T \bm{\eta}^{(k)} \|^2  
	 \geq \kappa_a^2 \| \bm{\eta}^{(k+1)} - \bm{\eta}^{(k)}\|^2. 
\end{aligned}
\end{split} 
\end{equation}
By combining \eqref{eq:bound_dual} and \eqref{eq:bound_dual_2}, we get
\begin{equation}
\label{eq:bound_dual_3}
\begin{split}
\begin{aligned}
 \| \bm{\eta}^{(k+1)} - \bm{\eta}^{(k)}\|^2  
&  \leq \frac{L_{\theta_1}^2}{\kappa_a^2}  	 \| \mathbf{x}^{(k+1)}  -   \mathbf{x}^{(k)} \|^2. 
\end{aligned}
\end{split} 
\end{equation}
Thus, for the $\bm{\eta}$-subproblem, we can use the primal variable to bound as follows: 
\begin{equation}
\label{eq:lang_3_a}
	\begin{split}
	\begin{aligned}
	& \mathcal{L}(\mathbf{x}^{(k+1)},\mathbf{w}^{(k+1)};\bm{\eta}^{(k+1)}) - 
	\mathcal{L}(\mathbf{x}^{(k+1)},\mathbf{w}^{(k+1)};\bm{\eta}^{(k)})  \\
&  = \langle  	\bm{\eta}^{(k+1)}, \mathbf{w}^{(k+1)} - \mathbf{\Omega} \mathbf{x}^{(k+1)}\rangle -
    \langle  	\bm{\eta}^{(k)}, \mathbf{w}^{(k+1)} - \mathbf{\Omega} \mathbf{x}^{(k+1)}\rangle  \\
&=  \frac{1}{\rho} \| \bm{\eta}^{(k+1)} -  \bm{\eta}^{(k)}\|^2  \leq \frac{L_{\theta_1}^2}{ \rho \kappa_a^2}  \| \mathbf{x}^{(k+1)}  -   \mathbf{x}^{(k)} \| ^2. 
	  \end{aligned}
	\end{split} 
	\end{equation} 
    Since the $\mathbf{x}$-subproblem is $\mu_x$-strongly convex we have 
     \begin{equation}
     	\begin{split}
  \label{eq:strong_convex}
 &\mathcal{L}(\mathbf{x}^{(k+1)},\mathbf{w}^{(k)};\bm{\eta}^{(k)}) \\
& \leq \mathcal{L}(\mathbf{x}^{(k)},\mathbf{w}^{(k)};\bm{\eta}^{(k)}) 
 - \frac{\mu_x}{2} \| \mathbf{x}^{(k+1)}  - \mathbf{x}^{(k)} \|^2.
	\end{split} 
\end{equation} 
Similarly, since the $\mathbf{w}$-subproblem is convex, we have the following inequality:
   	\begin{equation}
	\label{eq:first_order_w}
	\begin{split}
	\begin{aligned}
	\mathcal{L}(\mathbf{x}^{(k+1)},\mathbf{w}^{(k+1)};\bm{\eta}^{(k)})  \leq  \mathcal{L}(\mathbf{x}^{(k+1)},\mathbf{w}^{(k)};\bm{\eta}^{(k)}) .
		\end{aligned}
	\end{split} 
	\end{equation}
Thus, by combining \eqref{eq:lang_3_a}, \eqref{eq:strong_convex}, and \eqref{eq:first_order_w}, we obtain:
	\begin{equation}\label{eq:q_case_a}
	\begin{split}
	\begin{aligned}
	&\mathcal{L}(\mathbf{x}^{(k+1)},\mathbf{w}^{(k+1)};\bm{\eta}^{(k+1)})  -\mathcal{L}(\mathbf{x}^{(k)},\mathbf{w}^{(k)};\bm{\eta}^{(k)}) \\
	&= \mathcal{L}(\mathbf{x}^{(k+1)},\mathbf{w}^{(k+1)};\bm{\eta}^{(k+1)}) - 
	\mathcal{L}(\mathbf{x}^{(k+1)},\mathbf{w}^{(k+1)};\bm{\eta}^{(k)})  \\
	&   \quad + \mathcal{L}(\mathbf{x}^{(k+1)},\mathbf{w}^{(k+1)};\bm{\eta}^{(k)})  -
	\mathcal{L}(\mathbf{x}^{(k+1)},\mathbf{w}^{(k)};\bm{\eta}^{(k)})  \\
	& \quad + \mathcal{L}(\mathbf{x}^{(k+1)},\mathbf{w}^{(k)};\bm{\eta}^{(k)})   -
	  \mathcal{L}(\mathbf{x}^{(k)},\mathbf{w}^{(k)};\bm{\eta}^{(k)})  \\
	&\leq \frac{L_{\theta_1}^2}{ \rho \kappa_a^2}  \| \mathbf{x}^{(k+1)}  -   \mathbf{x}^{(k)} \| ^2 - \frac{\mu_x}{2} \| \mathbf{x}^{(k+1)}  - \mathbf{x}^{(k)} \|^2\\ 
	& = \left (\frac{L_{\theta_1}^2}{ \rho \kappa_a^2}  - \frac{\mu_x}{2} \right) \| \mathbf{x}^{(k+1)}  -   \mathbf{x}^{(k)} \| ^2,
		\end{aligned}
	\end{split} 
	\end{equation}
which will be negative provided by $\rho > {2L_{\theta_1}^2}/{\kappa_a^2 \mu_x}$. Thus, when $\rho > \max\left( \frac{2 L_{\theta_1}^2}{\kappa_a^2 \mu_x},\rho_0\right)$, the result follows.

Next, we prove \textbf{Case (b)}.  In this case, we do not assume convexity of the $\mathbf{x}$-subproblem. For the $\bm{\eta}$-subproblem, we obtain 
\begin{equation}
\label{eq:lang_3}
	\begin{split}
	\begin{aligned}
	& \mathcal{L}(\mathbf{x}^{(k+1)},\mathbf{w}^{(k+1)};\bm{\eta}^{(k+1)}) - 
	\mathcal{L}(\mathbf{x}^{(k+1)},\mathbf{w}^{(k+1)};\bm{\eta}^{(k)})  \\
&=  \frac{1}{\rho} \| \bm{\eta}^{(k+1)} -  \bm{\eta}^{(k)}\|^2.
	  \end{aligned}
	\end{split} 
	\end{equation} 
Let $\mathbf{\Omega}$ be full-column rank with $\mathbf{\Omega} ^\T \mathbf{\Omega} \succeq \kappa_b^2 \mathbf{I}$, which gives
\begin{equation}\label{eq:bound_dual_b}
\begin{split}
\begin{aligned}
	 & \|  \mathbf{\Omega} \bm{x}^{(k+1)}  -  \mathbf{\Omega}\bm{x}^{(k)} \|^2  
	 \geq \kappa_b^2 \| \mathbf{x}^{(k+1)} - \mathbf{x}^{(k)}\|^2.
\end{aligned}
\end{split} 
\end{equation}
 For the $\mathbf{x}$-subproblem we get
 \begin{equation}  \label{eq:lang_1}
\begin{split}
&  \mathcal{L}(\mathbf{x}^{(k)},\mathbf{w}^{(k)};\bm{\eta}^{(k)}) - 
\mathcal{L}(\mathbf{x}^{(k+1)},\mathbf{w}^{(k)};\bm{\eta}^{(k)})  \\
& =
 \theta_1(\mathbf{x}^{(k)}) - \theta_1(\mathbf{x}^{(k+1)}) + 
   \langle \bm{\eta}^{(k)},  \mathbf{\Omega}\mathbf{x}^{(k+1)}- \mathbf{\Omega}\mathbf{x}^{(k)} \rangle \\
   &\quad +
    \langle \rho(\mathbf{w}^{(k)} - \mathbf{\Omega}\mathbf{x}^{(k+1)} ),  \mathbf{\Omega}\mathbf{x}^{(k+1)}- \mathbf{\Omega}\mathbf{x}^{(k)} \rangle 
    \\
    &\quad + \frac{\rho}{2} \| \mathbf{\Omega}\mathbf{x}^{(k+1)}  - \mathbf{\Omega}\mathbf{x}^{(k)}\|^2 \\ 
    & \overset{\eqref{eq:first_order_x}} {=}  
   \theta_1(\mathbf{x}^{(k)}) - \theta_1(\mathbf{x}^{(k+1)}) + \frac{\rho}{2} \| \mathbf{\Omega}\mathbf{x}^{(k+1)}- \mathbf{\Omega}\mathbf{x}^{(k)}\|^2 \\ 
   &\quad +  \langle - \nabla \theta_1(\mathbf{x}^{(k+1)}),  \mathbf{x}^{(k)}- \mathbf{x}^{(k+1)} \rangle \\
& {\geq}  - \frac{L_{\theta_1}}{2} \|\mathbf{x}^{(k)} - \mathbf{x}^{(k+1)} \|^2 + 
     \frac{\rho}{2} \| \mathbf{\Omega}\mathbf{x}^{(k+1)}- \mathbf{\Omega}\mathbf{x}^{(k)}\|^2 \\
& \overset{\eqref{eq:bound_dual_b}}{\geq} 
 \left ( \frac{\rho  \kappa_b^2  }{2}  - \frac{L_{\theta_1}}{2}  \right)
    \| \mathbf{x}^{(k+1)}- \mathbf{x}^{(k)}\|^2,
\end{split} 
\end{equation}
and by combining \eqref{eq:first_order_w}, \eqref{eq:lang_3}, and \eqref{eq:lang_1}, we get
	\begin{equation} \label{eq:q_case_b}
	\begin{split}
	\begin{aligned}
	& \mathcal{L}(\mathbf{x}^{(k+1)},\mathbf{w}^{(k+1)};\bm{\eta}^{(k+1)})  -\mathcal{L}(\mathbf{x}^{(k)},\mathbf{w}^{(k)};\bm{\eta}^{(k)}) \\
	&= \mathcal{L}(\mathbf{x}^{(k+1)},\mathbf{w}^{(k+1)};\bm{\eta}^{(k+1)}) - 
	\mathcal{L}(\mathbf{x}^{(k+1)},\mathbf{w}^{(k+1)};\bm{\eta}^{(k)})  \\
	&   \quad + \mathcal{L}(\mathbf{x}^{(k+1)},\mathbf{w}^{(k+1)};\bm{\eta}^{(k)})  -
	\mathcal{L}(\mathbf{x}^{(k+1)},\mathbf{w}^{(k)};\bm{\eta}^{(k)})  \\
	& \quad + \mathcal{L}(\mathbf{x}^{(k+1)},\mathbf{w}^{(k)};\bm{\eta}^{(k)})   -
	  \mathcal{L}(\mathbf{x}^{(k)},\mathbf{w}^{(k)};\bm{\eta}^{(k)})  \\
	&\leq  \frac{ L_{\theta_1} - \rho  \kappa_b^2}{2} 
    \| \mathbf{x}^{(k+1)}- \mathbf{x}^{(k)}\|^2  +  \frac{1}{\rho} \| \bm{\eta}^{(k+1)} -  \bm{\eta}^{(k)}\|^2 ,
		\end{aligned}
	\end{split} 
	\end{equation}
which will be nonnegative provided that $\rho > \frac{L_{\theta_1} }{  \kappa_b^2 }$.

\section{Proof of Lemma 2}
\label{pf:lemma_gn}
The error between the local iterate $\mathbf{x}^{(i+1)}$ in the update and the minimizer $\mathbf{x}^{\star}$ satisfies the following recursion:
\begin{equation}
\label{eq:taylorbased_equality}
\begin{aligned}
&  \| \mathbf{x}^{(i+1)} - \mathbf{x}^\star \|  \\
& = \left\| [\mathbf{J}^\T \mathbf{J}(\mathbf{x}^{(i )}) ] ^{-1}  \right\|  \\
&  \quad \times  \left\| \mathbf{J}^\T \mathbf{J}(\mathbf{x}^{(i )}) (\mathbf{x}^{(i )} - \mathbf{x}^\star) -  \left[\nabla f(\mathbf{x}^{(i)})-\nabla f(\mathbf{x}^\star) \right] \right\| \\
& \leq \left\| [\mathbf{J}^\T \mathbf{J}(\mathbf{x}^{(i )}) ] ^{-1}  \right\| \\
& \quad \times 
\int_0^1 \left\| {
\mathbf{J}^\T \mathbf{J}(\mathbf{x}^{(i )}) - 
 ( \mathbf{J}^\T \mathbf{J} (\mathbf{x}^\star + \alpha (\mathbf{x}^{(i)} - \mathbf{x}^\star ) ) } \right. \\
 & \quad \quad \left. { 
 + \mathbf{H} (\mathbf{x}^\star + \alpha (\mathbf{x}^{(i)} - \mathbf{x}^\star ) )  ) } \right\| \, \left\| \mathbf{x}^{(i)} - \mathbf{x}^\star \right\| \dif\alpha
 \\
& \leq \left\| [\mathbf{J}^\T \mathbf{J}(\mathbf{x}^{(i )}) ] ^{-1}  \right\| \\
 & \quad \times 
  \int_0^1 \left\| \nabla^2 f(\mathbf{x}^{(i)}) - \nabla^2 f(\mathbf{x}^\star + \alpha (\mathbf{x}^{(i)} - \mathbf{x}^\star ) )
   - \mathbf{H} (\mathbf{x}^{(i )}) \right\| \\
    & \quad \quad \times  \left\| \mathbf{x}^{(i)} - \mathbf{x}^\star \right\|   \dif \alpha \\
& \leq \frac{L_f}{2}  \left\| [\mathbf{J}^\T \mathbf{J}(\mathbf{x}^{(i )}) ] ^{-1}  \right\| 
 \left\| \mathbf{x}^{(i)} - \mathbf{x}^\star  \right\|^2  \\
  & \quad + 
  \left\|  [\mathbf{J}^\T \mathbf{J}(\mathbf{x}^{(i )}) ] ^{-1} \mathbf{H} (\mathbf{x}^{(i )}) \right\|
    \left\| \mathbf{x}^{(i)} - \mathbf{x}^\star \right\|.
 \end{aligned}
\end{equation}
Let $\mathbf{H}(\mathbf{x})$ be bounded by $\epsilon_h$, that is, $\|\mathbf{H}(\mathbf{x}) \| \le \epsilon_h $.
We conclude that when $\epsilon_h \to 0$, the convergence is quadratic. Now let $\mathbf{J}^\T \mathbf{J}(\mathbf{x}^{(i )})  \succeq \mu^2 \mathbf{I}$.
Then, linear convergence is obtained when the following condition is satisfied:
\begin{equation}
\begin{split}
&\left\|  [\mathbf{J}^\T \mathbf{J}(\mathbf{x}^{(i )}) ] ^{-1} \mathbf{H} (\mathbf{x}^{(i )}) \right\| \\
&\qquad \leq  \left\|  [\mathbf{J}^\T \mathbf{J}(\mathbf{x}^{(i )}) ] ^{-1} \right\| \, \| \mathbf{H} (\mathbf{x}^{(i )}) \| 
\leq  \frac{\epsilon_h}{\mu^2} < 1.
\end{split}
\end{equation}

\bibliographystyle{IEEEtran}
\bibliography{refs}

\end{document}